%% file: main.tex
\documentclass[lettersize,journal]{IEEEtran}
\usepackage{amsmath,amsfonts}

\usepackage{algorithm}
\usepackage{array}
\usepackage[caption=false,font=normalsize,labelfont=sf,textfont=sf]{subfig}
\usepackage{textcomp}
\usepackage{stfloats}
\usepackage{url}
\usepackage{verbatim}
\usepackage{graphicx}
\usepackage{cite}
\usepackage{booktabs}
\usepackage{threeparttable}

\usepackage{caption}
\usepackage{algorithm,algpseudocode,tabularx}
\usepackage{xcolor}

\usepackage{acronym}
\input{Acronyms.tex}
\usepackage{dsfont}
\usepackage{amsthm}
\theoremstyle{plain}
\newtheorem{theorem}{Theorem}[section]

\newtheorem{lemma}[theorem]{Lemma}
\newtheorem{corollary}[theorem]{Corollary}
\theoremstyle{definition}

\newtheorem{assumption}[theorem]{Assumption}
\theoremstyle{remark}

\usepackage{tikz}
\usetikzlibrary{positioning, arrows.meta, fit, backgrounds, calc}
\usepackage{enumitem}
\usepackage{placeins}
\usepackage{pgfplots}
\pgfplotsset{compat=newest}

\begin{document}
\title{Beyond Labels: Information-Efficient Human-in-the-Loop Learning using Ranking and Selection Queries}

\author{Belén Martín-Urcelay, Yoonsang Lee, Matthieu R. Bloch, Christopher J. Rozell,
\thanks{B. Martín-Urcelay (e-mail: burcelay3@gatech.edu), Y. Lee, M. R. Bloch  and C. J. Rozell are with the School of Electrical and Computer Engineering, Georgia Institute of Technology, Atlanta, GA 30332 USA.}}

\maketitle

\begin{abstract}
Integrating human expertise into machine learning systems often reduces the role of experts to labeling oracles, a paradigm that limits the amount of information exchanged and fails to capture the nuances of human judgment. We address this challenge by developing a human-in-the-loop framework to learn binary classifiers with rich query types, consisting of item ranking and exemplar selection. We first introduce probabilistic human response models for these rich queries motivated by the relationship experimentally observed between the perceived implicit score of an item and its distance to the unknown classifier. Using these models, we then design active learning algorithms that leverage the rich queries to increase the information gained per interaction. We provide theoretical bounds on sample complexity and develop a tractable and computationally efficient variational approximation
. Through experiments with simulated annotators derived from crowdsourced word-sentiment and image-aesthetic datasets, we demonstrate significant reductions on sample complexity. We further extend active learning strategies to select queries that maximize information rate, explicitly balancing informational value against annotation cost. This algorithm in the word sentiment classification task reduces learning time by more than 57\% compared to traditional label-only active learning. 
\end{abstract}


\section{Introduction}
\label{sec:introduction}
Integrating machine learning systems with the nuanced judgments of human experts is a critical goal in domains ranging from image \cite{HelblingPrefGen:Attributes} and text \cite{Chen2025LearningModels} generation to word sense disambiguation \cite{Zhu2007ActiveProblem} and medical diagnosis \cite{Roy2022DemystifyingMedicine}. However, transferring human expertise to computational systems remains challenging. Although humans excel at making complex judgments, they often struggle to articulate the precise features or explicit logic that guide their decisions \cite{Casper2023}. This gap between intuitive expertise and machine-interpretable explanations largely reduces the role of humans to that of labeling oracles \cite{Canal2019, Shekhar2021ActiveAbstention}. As oracles, experts provide labels that enable supervised learning algorithms to approximate the implicit decision functions.

A drawback of this oracle-based paradigm is the substantial human effort and cost required to acquire large labeled datasets. Active learning \cite{Settles2009} has emerged as an efficient approach to mitigate this cost. Active learning strategies select informative examples to query, thereby significantly reducing the number of required labels. Despite these advances, the information obtained from such queries is constrained by the simplistic nature of labels \cite{Ashktorab2021AI-AssistedOverreliance}. In practice, queries are often limited to 
binary classifications, for which the information gained is inherently capped at just one bit \cite{Canal2021FeedbackLearning, Shekhar2020ActiveAbstention}. This information bottleneck raises the question: can we move beyond simple labeling to \emph{richer queries that provide more information per human interaction, while remaining intuitive for humans?}

Even when presented with binary labeling tasks (e.g., ``Is the weather today good or bad?''), humans may engage in richer cognitive processes. Humans often categorize items by recalling exemplars and comparing to reference points (e.g. ``This is the worst weather we have had all week''). Despite this capacity for comparative reasoning, most human-in-the-loop strategies artificially constrain experts to binary labeling roles. To harness the underlying context-rich evaluation, we consider alternative query types: ranking items by attribute strength or selecting the most representative exemplar from a list. Capitalizing on these richer query types requires quantitative models of human responses that render such queries machine-interpretable and optimizable,  as well as an algorithm that selects informative, cost-aware queries and learns from the corresponding human answers. We address these requirements through the following contributions:
\begin{itemize}
    \item \textbf{Human response models for rich queries on off-the-shelf embeddings.} We introduce probabilistic response models to ranking and exemplar-selection queries. These models are based on the observed relationship between the perceived score of an item and its distance to the decision boundary in the embedding space. This enables the use of rich queries that capture more information per interaction than traditional labeling. \looseness=-1
    \item  \textbf{Theoretical guarantees and empirical improvements in sample complexity.} We derive theoretical bounds for the expected stopping time. Empirical evaluations with simulated annotators based on human data demonstrate up to 85\% reduction in human interactions compared to traditional active labeling approaches. \looseness=-1
    \item \textbf{Cost-aware information-rate optimization with human timing models.} We formulate query selection as maximizing expected bits of information per second rather than bits per interaction. Using response time models derived from a crowdsourced study, our cost-aware approach reduces annotation time by more than half compared to label-only active learning on word sentiment classification tasks, demonstrating practical time savings beyond sample-complexity gains. \looseness=-1
\end{itemize}

Parts of this work have been previously presented at the Conference of Decision and Control \cite{Martin-Urcelay2024EnhancingClassification}. This version substantially extends it in four directions: We introduce ranking queries that gather full orderings with labels in a single interaction; we provide sample-complexity guarantees that explicitly quantify how query type and embedding dimension affect stopping time; we demonstrate generalization beyond word sentiment by adding experiments on image aesthetic classification; and we incorporate an information-rate objective together with empirically fitted human timing models to optimize time costs. Together, these additions turn the original proof-of-concept into a more theoretically grounded and cost-aware framework for rich-query human-in-the-loop learning. The paper is organized as follows. In Section~\ref{sec:method} we introduce the components of our method: human response models, an information-theoretic question selection strategy, and a tractable active learning algorithm for query item selection. In Sections~\ref{sec:theory} and ~\ref{sec:emp_sample_complexity} we present theoretical and empirical results on sample complexity. Finally, in Section~\ref{sec:emp_time} we describe crowdsourced experiments, derive human time cost models, and show expected time savings from our query-type selection based on information rate for the word sentiment classification task.

\section{Proposed Method}\label{sec:method}

    \begin{figure*}[htp]
        \centering
        \hspace*{-0.35cm}\begin{tabular}{cccc}
            \includegraphics[width=0.24\linewidth , clip, trim=3cm 9cm 3cm 9cm]{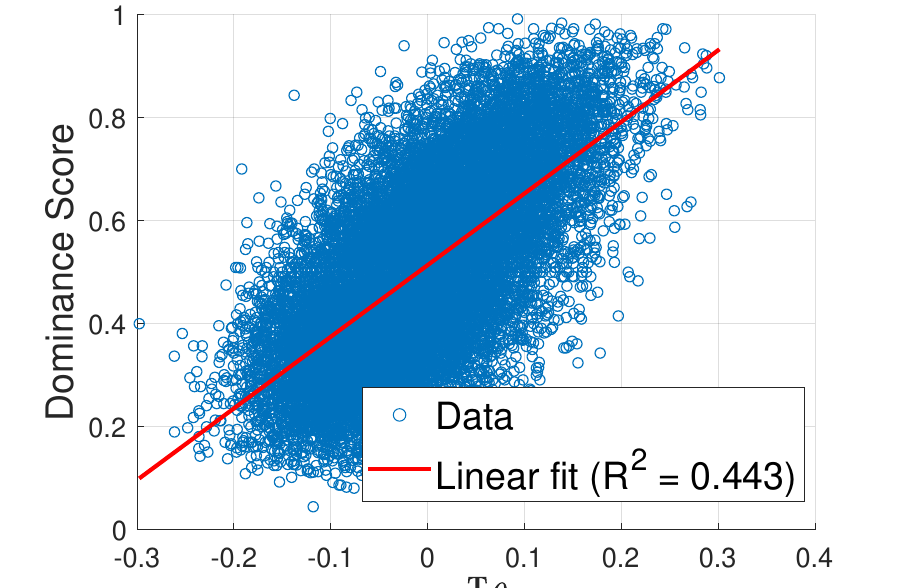} &
            \includegraphics[width=0.24\linewidth , clip, trim=3cm 9cm 3cm 9cm]{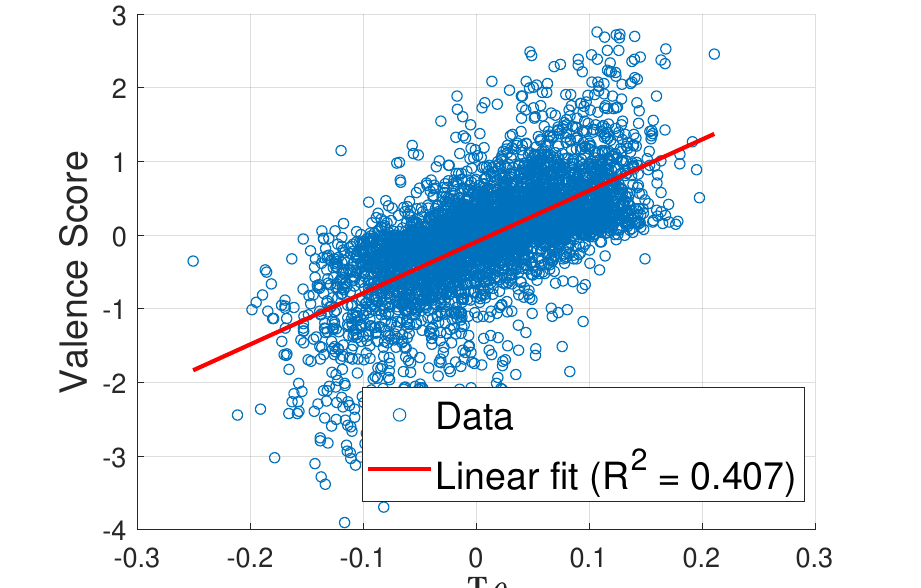} &
            \includegraphics[width=0.24\linewidth , clip, trim=3cm 9cm 3cm 9cm]{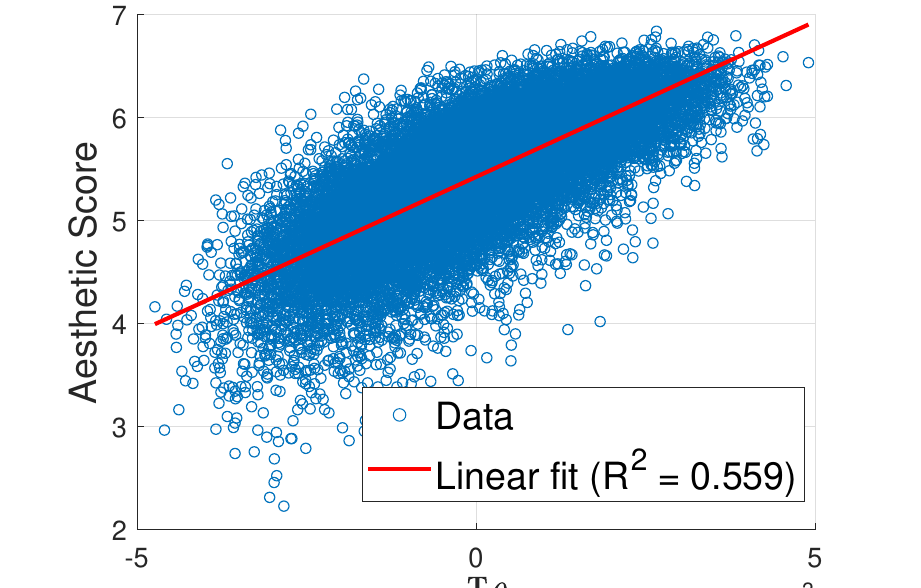} &
            \includegraphics[width=0.24\linewidth, clip, trim=3cm 9cm 3cm 9cm]{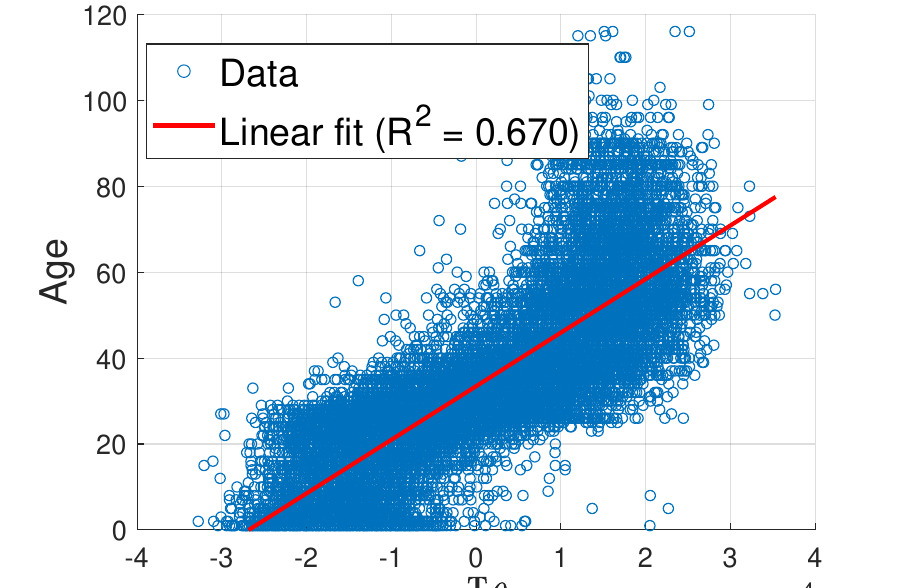} \\  
            a) Word2Vec embedding  & b) Word2Vec embedding  & c) ViT-L/14 embedding & d) InceptionResnetV1 \\
            \quad  for NRC dataset &  for SocialSent dataset & for AVA dataset &  for UTKFace dataset\\
        \end{tabular}
        \caption{Scores for words (a, b) and image (c, d) attributes as a function of the inner product between their pre-defined embedding and the ground truth classifier. We observe there exists an approximately affine relationship. Pre-trained embeddings naturally encode score information as distance from decision boundary, enabling information-rich queries beyond binary labels. }
        \label{fig.:score_vs_distance}
    \end{figure*}   
To leverage rich query types, we develop a quantitative framework comprising three components. First, we construct a human response model that predicts the likelihood of annotator answers to ranking and exemplar selection queries. Second, recognizing that different queries impose varying cognitive loads, we present an active learning algorithm that balances informational value against human effort. Third, we derive variational approximations and greedy heuristics that make Bayesian inference tractable in high-dimensional embedding spaces. Together, these components enable principled optimization over expressive query types while accounting for realistic cost constraints.
\subsection{Human Response Model}
\label{sec.:model}
    Our goal is to learn a decision boundary $\boldsymbol{\theta}$ in an embedding space that accurately predicts the implicit classification rules of human annotators. We focus on unitary norm binary linear classifiers $\boldsymbol\theta \in \mathbb{R}^{d}, \|\boldsymbol{\theta}\|_2 = 1$ that label every embedded item $\mathbf{x}\in \mathbb{R}^d$ in a $d$-dimensional embedding space as $y=\operatorname{sign}(\boldsymbol\theta^T\mathbf{x})$. These linear classifiers generalize to non-linear problems if there exists a mapping from the original non-linear space to a higher dimensional embedding in which the data is linearly separable.
    Given a set of items labeled by the annotator, a common method to learn a linear classifier is \emph{logistic regression}, in which the label probability is given by   
    \begin{equation} \label{eq.likelihood_label}
        \mathbb{P}\left[y=1|\mathbf{x} \right] = \left(1 + \exp{\left( w (\boldsymbol\theta^T \mathbf{x})\right)}\right)^{-1},
    \end{equation}
    with $w \in \mathbb{R}$ representing the inverse of the scale parameter. 
    
    The main drawback of reducing annotators to labelers is that each response conveys at most one bit of information. Consequently, accruing enough information to learn an accurate classifier often requires an impractically large number of queries in data-scarce settings. To mitigate this information bottleneck, we design queries that solicit information beyond binary labels. Not only must these queries be intuitive for human annotators to respond to, but the queries also require a response model that relates item embeddings to the classifier in a manner quantifiable for computational analysis. Our key observation is that off-the-shelf embeddings naturally exhibit geometry that aligns with perceived scores: We expect human annotators to be uncertain when classifying items whose embeddings lie close to the boundary between classes. Conversely, we expect annotators to become more confident as the items lie further away from the boundary, implicitly assigning higher and lower scores as we move in opposite directions from the classification boundary. Figure~\ref{fig.:score_vs_distance} demonstrates that this relationship holds across a variety of embeddings, tasks and datasets. Specifically, we examine popular word and image embeddings \cite{word2vec, Radford2021LearningSupervision, Cao2018VGGFace2:Age} for tasks ranging from word sentiment \cite{Hamilton} and dominance \cite{Mohammad2018} analysis, to image aesthetic perception \cite{Murray2012AVA:Analysis} and age categorization from a face \cite{Zhang2017AgeAutoencoder}. In all cases, \emph{we consistently observe a linear relationship between an item score and the inner product between its embedding and the \ac{MMSE} classifier}. We provide additional details in Appendix \ref{sec:app.LinearRelationship}. Unlike previous work \cite{Liu2018, Canal2019} that learns task-specific embedding spaces, we exploit this naturally occurring linear relationship. This empirical observation motivates us to formalize the relationship between embeddings and scores as follows; \looseness=-1    
    \begin{assumption} \label{assum_humanResponse}
       Annotators associate a score to an item with embedding $\mathbf{x}_i$ as 
        \begin{equation}\label{eq.linear_score}
            \text{score}(\mathbf{x}_i) = a  \mathbf{x}_i^T\boldsymbol\theta + b + \delta_i,
        \end{equation}
        where $\delta_i$ represents the noise associated with $\mathbf{x}_i$, with a query dependent extreme value distribution. The scalars $a$ and $b$, which describe the affine relationship, are dataset and attribute dependent.
    \end{assumption}       

    The score model enables us to derive probability distributions of human responses for a large variety of queries. We model the noise in Equation~\eqref{eq.linear_score} so that it leads to the Boltzmann choice model used in behavioral economics \cite{Georgii2011GibbsTransitions}. We validate the distribution choices empirically in Appendix~\ref{sec:app.Gumbel}. Consider the question $q_{\text{high}} = $ ``Select and label the item with highest score;'' and the noise distribution $\delta_i \sim \operatorname{Gumbel-Max}(\mu, \sigma)$. In that case, the probability that the $i$-th item is selected from a set $\mathcal{S} = \{\mathbf{x}_k\}_{k=1}^{|\mathcal{S}|}$ is
    \begin{align}\label{eq.likelihood_qhigh}
        \mathbb{P}[i | \mathcal{S}, \boldsymbol{\theta}, q_{\text{high}}] &= \mathbb{P}[ \text{score}(\mathbf{x}_i)>\text{score}(\mathbf{x}_j), \forall j\neq i] \nonumber\\
        & =\mathbb{P}[\delta_{j}-\delta_{i} <a  (\mathbf{x}_i-\mathbf{x}_j)^T\boldsymbol\theta, \forall j\neq i] \nonumber\\
        & = \frac{\exp{(\frac{a}{\sigma}\mathbf{x}_i^T\boldsymbol\theta)}}{\sum_{\mathbf{x}\in\mathcal{S}}\exp{(\frac{a}{\sigma}\mathbf{x}^T\boldsymbol\theta)}},
    \end{align}
    because this is a logit choice probability \cite{Train2003}. Analogously, for the question $q_{\text{low}} = $``Select and label the item with lowest score'' and noise distribution $\delta_i \sim \operatorname{Gumbel-Min}(\mu, \sigma)$, the probability that the annotator selects the $i$-th item is \looseness=-1
    \begin{align}\label{eq.likelihood_qlow}
        \mathbb{P}[i | \mathcal{S}, \boldsymbol{\theta}, q_{\text{low}}] 
        & = \frac{\exp{(-\frac{a}{\sigma}\mathbf{x}_i^T\boldsymbol\theta)}}{\sum_{\mathbf{x}\in\mathcal{S}}\exp{(-\frac{a}{\sigma}\mathbf{x}^T\boldsymbol\theta)}}.
    \end{align}    

    We extend the model to ranking queries $q_{\text{rank}} = $ ``Rank the items from highest to lowest score and indicate which is the last positive example in the ranked list.'' Given a set of items $\mathcal{S} = \{\mathbf{x}_k\}_{k=1}^{|\mathcal{S}|}$, let $\mathbf{r} = [r_1, ..., r_{|\mathcal{S}|}]$ represent the permutation that orders these items, where $\mathbf{x}_{r_j}$ is the item ranked at position $j$. We model the annotator's ranking as a sequential selection process:
    \begin{align}\label{eq.likelihood_qrank}
        \mathbb{P}\left[\mathbf{r} | \mathcal{S}, \boldsymbol{\theta}, q_{\text{rank}}\right] 
        & = \prod_{j=1}^{|\mathcal{S}|-1} \mathbb{P}\left[r_j | \mathcal{R}_j, \boldsymbol{\theta}, q_{\text{high}}\right]\nonumber\\
        & =\prod_{j=1}^{|\mathcal{S}|-1}  \frac{\exp{(\frac{a}{\sigma}\mathbf{x}_{r_j}^T\boldsymbol\theta)}}{\sum_{\mathbf{x}\in \mathcal{R}_j}\exp{(\frac{a}{\sigma}\mathbf{x}^T\boldsymbol\theta)}},
    \end{align}    
    where $\delta_i \sim \operatorname{Gumbel-Max}(\mu, \sigma)$ and $\mathcal{R}_j = \mathcal{S} \setminus \{ \mathbf{x}_{r_1}, ..., \mathbf{x}_{r_{j-1}}\} $. This response model is known as the Plackett-Luce model \cite{luce1959individual, plackett1975analysis}. By asking the annotators to mark the ``last positive example'' in the ordered list, we obtain a threshold $\ell\in\{0, 1, ..., |\mathcal{S}|\}$ that implicitly captures the labels for all items. Namely, if $\ell=0$ all items will receive a negative label. Otherwise, all items in positions $i\leq \ell$ receive a positive label, while the remaining items receive a negative label.

    The queries $q_{\text{high}} $  and $q_{\text{low}} $ request a listwise choice, gathering up to $\log_2 |\mathcal{S}|$ more bits of information per query than a traditional binary label. The ranking query $q_{\text{rank}} $ receives a full ordering over $\mathcal{S}$ and complete labeling further alleviating the information gain bottleneck. Table \ref{tab:query_comparison} summarizes the differences between the proposed rich queries and the traditional labeling query.

\begin{table*}[t]
\centering
\caption{Comparison of Query Types for Human-in-the-Loop Learning}
\label{tab:query_comparison}
\begin{tabular}{@{}lcccc@{}}
\toprule
\textbf{Query Type} & \textbf{Outcome} & \textbf{Information} & \textbf{Expected Response} & \textbf{Interactions to}  \\
                    & \textbf{Space $(N)$} & \textbf{per Query (bits)} & \textbf{Time (s)} & \textbf{75\% Accuracy} \\
\midrule
Label only & $2$ & $\leq 1$ & $4.37$ & $1310$  \\
\midrule
Label + Selection & $2|\mathcal{S}|$  & $ \leq 1 + \log_2|\mathcal{S}|$ & $4.01 + 0.63|\mathcal{S}|$ & $468$ \\
($q_{\text{high}}, q_{\text{low}}$) &  &($|\mathcal{S}|=4$: $\sim$3 bits) & ($|\mathcal{S}|=4$: 6.5s) & ($|\mathcal{S}|=4$) \\
\midrule
Label Threshold + Ranking& $(|\mathcal{S}|+1)!$  & $\leq \log_2(|\mathcal{S}+1|!)$ & $-0.32 + 4.41|\mathcal{S}|$ & 191 \\
($q_{\text{rank}}$) & & ($|\mathcal{S}|=4$: $\sim$4.3 bits) & ($|\mathcal{S}|=4$: 17.3s) & ($|\mathcal{S}|=4$)  \\
\bottomrule
\end{tabular}
\begin{tablenotes}
\small
\item Information per query represents the theoretical maximum mutual information $I(\theta; \mathcal{O}|q, S)$. Expected response times are modeled from crowdsourced experiments with human participants on word sentiment classification tasks (Section III-C). Performance metric (75\% accuracy) is based on word sentiment classification experiments with \textbf{active} word selection (Figure 3d) \looseness=-1.
\end{tablenotes}
\end{table*}

\subsection{Question Selection}\label{sec:question_selection}

 To minimize sample complexity, we wish to select the question and set size such that the information gained from the query $I(\boldsymbol\theta;o|q, |\mathcal{S}|)$ is maximized. In practice, feasible queries are often constrained by cognitive and interface limitations \cite{tiferet2025constraints}. Under our response model, ranking queries subsume exemplar selection queries; thus, they provide strictly more information. Therefore, to minimize sample complexity, when the domain allows it, we should select $q_{\text{rank}}$. We also empirically observe that actively selecting between the questions $q_{\text{high}}$ and $q_{\text{low}}$ does not have a significant impact on performance, so when asking selection queries, we alternate between both uniformly at random. Moreover, we find that the information gain is increasing with $|\mathcal{S}|$; thus, to minimize sample complexity, we select the largest feasible $|\mathcal{S}|$.

In many applications, however, the primary objective is to minimize a cost different from sample complexity, such as cognitive effort or total response time, and not all queries incur the same cost \cite{SettlesBurrandCravenMarkandFriedland2008, Kapoor2007SelectiveLearning, Canal2020}. To balance informativeness against response burden, we propose selecting queries that maximize the \emph{information rate}, i.e., the expected bits of information gained per unit cost:
 \begin{equation}\label{eq:rate}
     R = \frac{I(\boldsymbol\theta;o|q, |\mathcal{S}|)}{\mathbb{E}[\operatorname{cost}(q, |\mathcal{S}|)]}.
 \end{equation}
 Note that to maximize this rate, we need a model of the cost for each feasible question and item size combination, which we estimate in Section~\ref{sec:human_cost_model} for a word classification task. \looseness=-1

\subsection{Algorithm}\label{sec:algorithm}
We propose an online machine learning algorithm to learn a binary classifier from human feedback. Figure~\ref{fig.:BasicBlockDiagram} illustrates the principle of our approach. At each iteration $t$, the \emph{item selector} chooses the item set $\mathcal{S}$ such that the expected annotator response to the preselected question $q$ is as informative as possible, i.e., the set $\mathcal{S}$ that maximizes the mutual information between the underlying classifier and the annotator response. Next, the annotator answers the query with $o_t$. In the case of $q_{\text{high}}$ or $q_{\text{low}}$, the answer is an item and its corresponding label $o_t = (i_t, y_t)$; in the case of a ranking question $q_{\text{rank}}$, the answer involves an item ordering and threshold separating positive from negative items $o_t = (\mathbf{r}, \ell)$. The \emph{classifier estimator} collects the response from the annotator and leverages this information to update the estimator of the classifier $\boldsymbol{\theta}$. The posterior $\mathbb{P} [\boldsymbol{\theta}| \mathcal{F}_{t}]$, where $\mathcal{F}_t = \{o_k, q_k, \mathcal{S}_k\}_{k=1}^t$ denotes the history, is updated using Bayes rule to leverage the likelihood functions in Equations (\ref{eq.likelihood_label}), (\ref{eq.likelihood_qhigh}), (\ref{eq.likelihood_qlow}) and (\ref{eq.likelihood_qrank}). The algorithm continues querying the oracle until the uncertainty is sufficiently reduced. We measure uncertainty as the determinant of the posterior covariance matrix $|\boldsymbol{\Sigma}_t|$, which quantifies the volume of the uncertainty region, and we terminate when $|\boldsymbol{\Sigma}_t| \leq \epsilon^d$, where $\epsilon\in\mathbb{R}$ is the per dimension threshold. The approach is summarized in Algorithm~\ref{algo.:ideal}.

\begin{figure}[bh]
    \centering
    \begin{tikzpicture}[
        block/.style={rectangle, draw, thick, text width=2.5cm, minimum height=1cm, align=center},
        subblock/.style={rectangle, fill=white, draw, thick, text width=2.8cm, minimum height=1.9cm, align=center},
        circleblock/.style={circle, draw, thick, align=center, inner sep=2pt},
        node distance=0.3cm and 0.5cm,
    ]
    
        \node[subblock, fill=red!30] (ex-selector) {Item Selector \\ $\underset{\mathcal{S} 
        }{\mathrm{argmax}}  \ I(\boldsymbol{\theta};o|\mathcal{S}, q)$};
        \node[subblock, fill=red!30] (learner-estimator) [below=of ex-selector, yshift=-0.2cm] {Classifier Estimator \\ $p_{t} = \mathbb{P}[\boldsymbol\theta|\mathcal{F}_t]$};
    
        \begin{scope}[on background layer]
            \node[draw, dashed, red, fill=red!10, fit=(ex-selector) (learner-estimator), inner sep=0.3cm, label=above:Learning Algorithm] (ML_box) {};
        \end{scope}
    
       \node[right=2.6cm of ML_box.north east, anchor=north] (human-icon) {\includegraphics[scale=0.36]{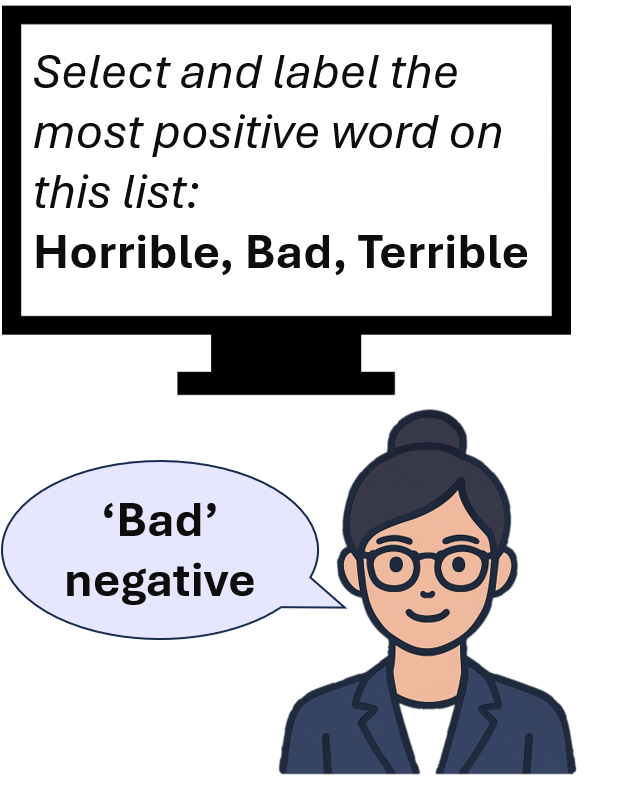}};
    
        \draw[->, thick] (ex-selector.east) -- ++ (0.6,0) -- node[above, pos=0.1] {$
        \mathcal{S}_t$} ($(human-icon.west) + (0.1,1.32)$);
        
         \draw[->, thick] ($(human-icon.west) + (0.1,-1.2)$) -- node [above, pos=0.35] {$o_t$} ($(learner-estimator.east)+ (0,-0.1)$);
        
        \draw[->, thick] (learner-estimator.north) -- ++(0,0.2) -| (ex-selector.south)
            node[near start, left] {$p_t$};
    
        \node[circleblock, left=of ex-selector] (theta) {$\mathcal{X}$};
        \node[circleblock, left=of learner-estimator] (prior) {$p_0$};

        \draw[->, thick] (theta.east) -- (ex-selector.west);
        \draw[->, thick] (prior.east) -- (learner-estimator.west);
    
    \end{tikzpicture}
    \caption{Block diagram for human-in-the-loop learning for sentiment word classification. At each interaction, the human annotator receives the query with items that maximize the information gain about the ground truth classifier $\boldsymbol{\theta}$. In the example, we ask the annotator to select a word from a list, and provide its label. The answer to the query is used to update the estimator of the classifier and select the next query items.}
    \label{fig.:BasicBlockDiagram}
\end{figure}
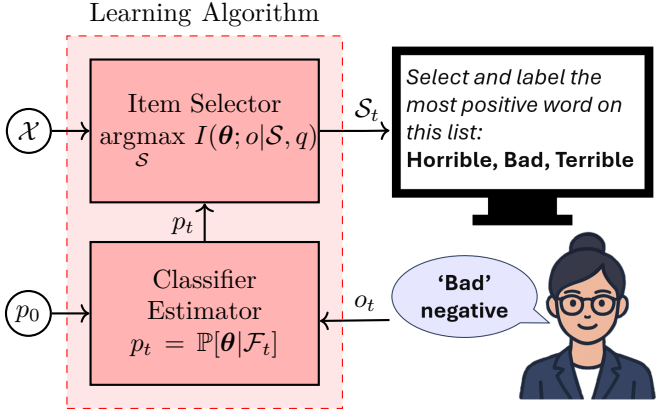

    \begin{algorithm}
        \caption{Ideal Human-in-the-Loop Learning (HiLL)} 
        \label{algo.:ideal}
        \begin{algorithmic}[1]
            \State \textbf{Input:} $\mathcal{X},  \mathbf{q}, |\mathcal{S}|, \mathbb{P} [\boldsymbol{\theta}| \mathcal{F}_{0}] , \epsilon^d$
            \State $t = 0$
            \While {$\left|\boldsymbol\Sigma_{t}\right| > \epsilon^d $}
            \State $\mathcal{S}_t \leftarrow \underset{
            \mathcal{S} \in \binom{|\mathcal{X}|}{|\mathcal{S}|}}{ \operatorname{argmax}}\  \mathbb{E
            }
            \left[I(\boldsymbol\theta;o|q_t, \mathcal{S}, \mathcal{F}_{t-1}) \right]$ 
            \State $o_t \leftarrow$ human's response to the query 
            \State $\mathbb{P} [\boldsymbol{\theta}| \mathcal{F}_{t}] = \frac{ \mathbb{P}\left[o_t| \boldsymbol\theta, q_t, \mathcal{S}_t \right]  \mathbb{P} [\boldsymbol{\theta}| \mathcal{F}_{t-1}]}{\mathbb{P}\left[o_t| q_t, \mathcal{S}_t \right]} $ 
            \State $t =t+1$
            
        \EndWhile
        \end{algorithmic} 
    \end{algorithm}

Unfortunately, Algorithm~\ref{algo.:ideal} is intractable in high dimensional settings, because the Bayesian update lacks a closed-form solution and the set of possible items grows combinatorially. The next subsections provide approximations to make the computations feasible. The corresponding tractable implementations are described in Algorithm~\ref{algo.:approx_selection} and Algorithm~\ref{algo.:approx_ranking}.

\subsubsection{Approximation of Belief Update}
    \begin{algorithm}[btp]
        \caption{Approximate HiLL for Selection} 
        \label{algo.:approx_selection}
        \begin{algorithmic}[1]
            \State \textbf{Input:} $\mathcal{X},|\mathcal{S}|, \boldsymbol{\mu}_0, \Sigma_0, \epsilon^d, N$ 
            \State $t = 1$
            \While {$\left|\boldsymbol\Sigma_{t}\right| > \epsilon^d $}
            \State $q_t \leftarrow$ sample uniformly from $\{q_{\text{high}}, q_{\text{low}}\}$            
            \State $S_t \leftarrow $item\_set\_selection$( q_t, \mathcal{X}, |\mathcal{S}|, \boldsymbol{\mu}_{t-1}, \Sigma_{t-1} , N)$     
            
            \State $i_t, y_t \leftarrow$ human response to the query 
            \State $\boldsymbol\mu_t, \Sigma_t \leftarrow$ belief\_update$(\mathcal{S},i_t, y_t, \boldsymbol{\mu}_{t-1}, \boldsymbol{\Sigma}_{t-1})$
            
            \State $t = t + 1$
            
        \EndWhile
        \end{algorithmic} 
    \end{algorithm}

    \begin{algorithm}[btp]
        \caption{Approximate HiLL for Ranking} 
        \label{algo.:approx_ranking}
        \begin{algorithmic}[1]
            \State \textbf{Input:} $\mathcal{X}, |\mathcal{S}|, \boldsymbol{\mu}_0, \Sigma_0, \epsilon^d, N$ 
            \State $t = 1$
            \While {$\left|\boldsymbol\Sigma_{t}\right| > \epsilon^d $}
            \State $S_t \leftarrow $item\_set\_selection$(q_{\text{rank}}, \mathcal{X}, |\mathcal{S}|, \boldsymbol{\mu}_{t-1}, \Sigma_{t-1}, N)$             
            \State $\mathbf{r}_t, l_t \leftarrow$ human response to the query 
            \For{$i = 1, 2, ..., |\mathcal{S}|$}
                \State $y \leftarrow 1$ if $i\leq l$, else $-1$
                \State $\boldsymbol\mu_t, \Sigma_t \leftarrow$ belief\_update$(\mathcal{S},r_i, y, \boldsymbol{\mu}_{t-1}, \boldsymbol{\Sigma}_{t-1})$
                
            \EndFor
            
            \State $t = t + 1$
            
        \EndWhile
        \end{algorithmic} 
    \end{algorithm}

    \begin{algorithm}[btp]
        \caption{belief\_update} 
        \label{algo.:belief_update_approx}
        \begin{algorithmic}[1]
            \State \textbf{Input:} $\mathcal{S},i_t, y_t, \boldsymbol{\mu}_{t-1}, \boldsymbol{\Sigma}_{t-1}, w, K$
            \State $\boldsymbol\mu_t, \Sigma_t \leftarrow \boldsymbol\mu_{t-1}, \Sigma_{t-1} $
            \While{$\boldsymbol{\mu}_t, \Sigma_t$ not converged}
                \While{$\boldsymbol{\mu}_t, \Sigma_t$ not converged}
                    \State $\xi^2 = w^2\mathbf{x}^T\boldsymbol\Sigma_{t}\mathbf{x} + w^2(\mathbf{x}^T\boldsymbol\mu_{t})^2$
                    \State $\boldsymbol\Sigma_{t}^{-1} =\boldsymbol\Sigma_{t-1}^{-1}+2 \frac{\tanh(\xi/2)}{4\xi} w^2\mathbf{x}\mathbf{x}^T$
                    \State $\boldsymbol\mu_{t}  =\boldsymbol\Sigma_{t}\left[\boldsymbol\Sigma_{t-1}^{-1} \boldsymbol\mu_{t-1}+\left(y-\frac{1}{2}\right) w\mathbf{x}\right]$
                \EndWhile

                \State $\boldsymbol{\mu}_t, \Sigma_t \leftarrow \underset{\boldsymbol\mu_q, \boldsymbol\Sigma_q}{\mathrm{argmin}}  \operatorname{KL}\left(\mathcal{N}(\boldsymbol\mu_q, \boldsymbol\Sigma_q)||\mathcal{N}(\boldsymbol\mu_t, \boldsymbol\Sigma_t)\right) $                
                \Statex \quad \quad \quad \quad \quad \quad \quad  $- K\mathbf{x}_i^T \boldsymbol\mu_q  $
                \Statex \quad \quad \quad \quad \quad \quad $\ \  + \log {\sum_{j=1}^{|\mathcal{S}|}\exp \left(K \mathbf{x}_j^T\boldsymbol \mu_q + \frac{1}{2} \mathbf{x}_j^T \boldsymbol\Sigma_q \mathbf{x}_j\right)}$

            \EndWhile
            \State \textbf{Output:} $\boldsymbol{\mu}_t, \Sigma_t $
        \end{algorithmic} 
    \end{algorithm}

    \begin{algorithm}[btp]
        \caption{item\_set\_selection} 
        \label{algo.:item_set_selection}
        \begin{algorithmic}[1]
            \State \textbf{Input:} $q, \mathcal{X}, |\mathcal{S}|, \boldsymbol{\mu}, \Sigma, N,  w, K$ 
            \State $\mathcal{S} \leftarrow \{\}$ 
            \For{$i = 1, 2, ..., |\mathcal{S}|$}
                \State $\{\widehat{\boldsymbol{\theta}}_n\}_{n=1}^N \leftarrow$ sample i.i.d. from $\mathcal{N}(\boldsymbol{\mu}, \Sigma)$
                
                \State $\mathbf{s} \leftarrow $ select item from dataset with Eq. (\ref{eq.word_selection})
                \State $\mathcal{S} \leftarrow \{\mathcal{S} , \mathbf{s}\}$ 
            \EndFor
            \State \textbf{Output:} $\mathcal{S}$
            
        \end{algorithmic} 
    \end{algorithm}

As line 6 of Algorithm~\ref{algo.:ideal} indicates, we use a Bayesian approach to update the belief of the classifier given the latest observation. When the question is $q_t \in \{ q_{\text{high}}, q_{\text{low}}\} $ and the answer is $o_t = (i_t, y_t)$, the posterior is given by
\begin{align*}
    \mathbb{P} [\boldsymbol{\theta}| \mathcal{F}_{t}]= \frac{\mathbb{P}[i_t| \boldsymbol{\theta},  \mathcal{S}_t, q_t] \mathbb{P}[y| \mathbf{x}_{i}, \boldsymbol{\theta}]}{\mathbb{P}[i_t, y_t| \mathcal{S}_t, q_t, \mathcal{F}_{t-1}]} \mathbb{P}[\boldsymbol{\theta}| \mathcal{F}_{t-1}],
\end{align*}
where $\mathcal{F}_0$ is the empty set $\emptyset$.

The likelihood functions are not conjugates of the prior, so no analytical closed form expression exists to compute the posterior. Although, \ac{BBVI} \cite{Ranganath2014BlackInference} is commonly used to approximate the posterior, our closed-form derivation of the variational updates for this setting, which we describe next, provides a lower variance and computationally cheaper approximation.

We approximate the classifier's density function as a Gaussian distribution $\boldsymbol{\theta} \sim \mathcal{N}(\boldsymbol\mu, \boldsymbol{\Sigma})$. We may then compute the posterior mean and variance given an item label by an iterative process \cite{Jaakkola2000} described in lines 4 to 7 of Algorithm~\ref{algo.:belief_update_approx}.
In a similar fashion, we approximate the classifier posterior given the item selected with a variational approach. We look for the variational distribution $q(\boldsymbol\theta) \sim \mathcal{N}(\boldsymbol\mu_q, \boldsymbol\Sigma_q)$ closest, in terms of the \ac{KL} distance, to the true posterior. This is equivalent to finding the distribution that maximizes the \ac{ELBO} \cite{KevinP.Murphy2022ProbabilisticTopics}
\begin{multline} \label{eq.ELBO}
    \text{ELBO}(q) 
 =   - \operatorname{KL}(q(\boldsymbol\theta) \| p(\boldsymbol\theta)) + \mathbb{E}_{\boldsymbol\theta \sim q} \left[K\mathbf{x}_i^T \boldsymbol\theta \right]  \\ -\mathbb{E}_{\boldsymbol\theta \sim q} \left[\log {\sum_{j=1}^{|\mathcal{S}|}\exp \left(K\mathbf{x}_j^T \boldsymbol\theta\right)} \right],
\end{multline}
where $\mathbf{x}_i$ is the embedding of the item selected by the human, $K = \frac{a}{\sigma}$ for $q_{\text{pos}}$ queries or $K = -\frac{a}{\sigma}$ for $q_{\text{neg}}$ queries, and the prior is $p(\boldsymbol\theta)\sim \mathcal{N}(\boldsymbol\mu_p, \boldsymbol\Sigma_p)$.
The first term in Equation (\ref{eq.ELBO}) is the \ac{KL} divergence between two Gaussian distributions,
\begin{align*}
    \operatorname{KL}(q||p) =&\frac{1}{2}\left[\log\frac{|\Sigma_p|}{|\Sigma_q|} -d+ \boldsymbol{\mu_q}^T\Sigma_p^{-1}\boldsymbol{\mu_q} +\boldsymbol{\mu_p}^T\Sigma_p^{-1}\boldsymbol{\mu_p}\right]\\ &  -\boldsymbol{\mu_q}^T\Sigma_p^{-1}\boldsymbol{\mu_p} + \frac{1}{2}tr\left\{\Sigma_p^{-1}\Sigma_q\right\}.
\end{align*}
Because of the linearity property of the expectation, we compute the second term in Equation (\ref{eq.ELBO}) as
\begin{equation*}
    \mathbb{E}_{\boldsymbol\theta \sim q} \left[K\mathbf{x}_i^T \boldsymbol\theta \right] = K\mathbf{x}_i^T \boldsymbol\mu_q.
\end{equation*}
The third term in Equation~(\ref{eq.ELBO}) has no closed-form solution, but following \cite{Braun2007VariationalChoice}, we apply Jensen's inequality to obtain an upper bound
\begin{multline*}
    \mathbb{E}_{\boldsymbol\theta \sim q} \left[\log {\sum_{j=1}^{|\mathcal{S}|}\exp \left(K \mathbf{x}_j^T\boldsymbol \theta\right)} \right]\\
     \leq \log {\sum_{j=1}^{|\mathcal{S}|}\exp \left(K \mathbf{x}_j^T\boldsymbol \mu_q + \frac{1}{2} K^2\mathbf{x}_j^T \boldsymbol\Sigma_q \mathbf{x}_j\right)}.
\end{multline*}
We approximate the posterior distribution given the item selection as a Gaussian distribution whose mean and covariance are obtained by maximizing the \ac{ELBO} lower bound
\begin{align} \label{eq.update_given_item}
    \{\boldsymbol\mu, \boldsymbol\Sigma\} 
    & = \underset{\boldsymbol\mu_q, \boldsymbol\Sigma_q}{\mathrm{argmin}}  \operatorname{KL}(q||p) - K\mathbf{x}_i^T \boldsymbol\mu_q   \nonumber\\
    & \quad + \log {\sum_{j=1}^{|\mathcal{S}|}\exp \left(K \mathbf{x}_j^T\boldsymbol \mu_q + \frac{1}{2} K^2 \mathbf{x}_j^T \boldsymbol\Sigma_q \mathbf{x}_j\right)}.
\end{align}
Putting these together, Algorithm~\ref{algo.:belief_update_approx} approximates the posterior by accounting for both the label and item selection. We first update the posterior according to the label received. Then we update the posterior according to the selected item with Equation (\ref{eq.update_given_item}). We repeat both updates until convergence. 

When $q_t = q_{\text{rank}}$, we compute the posterior as a recursion of $q_{\text{high}}$ queries such that
\begin{multline*}
\mathbb{P}[\boldsymbol{\theta}\mid \mathcal{F}_{t}] =
\prod_{j=1}^{|\mathcal{S}|}
\frac{
\mathbb{P}[r_j\mid \boldsymbol{\theta},\mathcal{R}_j,q_{\text{high}}]\,
\mathbb{P}[y_t\mid \mathbf{x}_{r_j},\boldsymbol{\theta}]
}{
\mathbb{P}\bigl[r_j,y_j \bigm| \mathcal{R}_j,\allowbreak q_{\text{high}},\allowbreak \mathcal{F}_{t-1},\allowbreak \{(r_k,y_k)\}_{k=1}^{j-1}\bigr]
}
\\
{}\times \mathbb{P}[\boldsymbol{\theta}\mid \mathcal{F}_{t-1}],
\end{multline*}
where $y_j= \mathds{1}[j\leq l_t]$. Said differently, Algorithm~\ref{algo.:belief_update_approx} is applied recursively, starting from the top item in the ranked list.

\subsubsection{Active Learning Heuristic for Item Set Selection}\label{sec:AL_heuristic}
Active learning \cite{Settles2012ActiveLearning, Shekhar2021ActiveAbstention} looks for the most informative items for human annotation, such that the sample complexity is minimized. This implies querying about the items that provide the most information about the ground truth on expectation. However, this maximization, as defined in line 4 of Algorithm~\ref{algo.:ideal}, requires computing the posterior over every possible item set and annotator response. There are combinatorially many options to compare, so even using the belief approximation, this approach is often computationally intractable.

To select the items in the query, we approximate the information gain based on query by committee \cite{KachitesMcCallum1998EmployingClassification}. At each iteration $t$, we sample $N$ particles $\widehat{\boldsymbol\theta}_n \overset{\text{i.i.d}}{\sim} \mathbb{P}[\boldsymbol\theta|\mathcal{F}_{t}]$. We maximize the disagreement between the prediction of each particle and the mean prediction among all particles as,
\begin{align*}
   \mathcal{S}_t = \underset{\mathcal{S} \in \binom{|\mathcal{X}|}{|\mathcal{S}|}}{\operatorname{argmax}}\quad & H\left[\frac{1}{N}\sum_{n=1}^Np_n(o| \mathcal{S})\right]\\ & - \frac{1}{N}\sum_{n=1}^NH\left[ p_n(o| \mathcal{S})\right],
\end{align*}
where $p_n(o| \mathcal{S}):=\mathbb{P}\left[o_t=o \mid \widehat{\boldsymbol{\theta}}_n, q_t, \mathcal{S}\right]$ represents the probability mass function over answers $o \in \mathcal{O}$ for query $q_t$ conditioned on the drawn classifier $ \widehat{\boldsymbol{\theta}}_n$, and $H[p]:= -\sum_{o \in \mathcal{O}} p(o) \log_2 p(o)$ is the Shannon entropy.
The first term in the objective function promotes queries with a high uncertainty of the expected output, which avoids queries for which the answer is predictable and thus not very informative. The second term attempts to minimize uncertainty due to intrinsic noise, for example, discouraging asking labels of neutral words such as ``table'' for which the uncertainty mostly comes from the labeling noise from humans, and not from a lack of exploration.

Actively selecting the item set significantly improves the performance, but there exist combinatorially many sets $\binom{|\mathcal{X}|}{|\mathcal{S}|}$ over which to maximize. To avoid this computational burden, we greedily aggregate a single item
\begin{align} \label{eq.word_selection}
  \mathbf{s} = \underset{\mathbf{s} \in \mathcal{X}}{\operatorname{argmax}}\quad & H\left[\frac{1}{N}\sum_{n=1}^N p_n\left(o| \{\mathcal{S}, \mathbf{s}\}\right)\right]\nonumber\\ & - \frac{1}{N}\sum_{n=1}^NH\left[p_n\left(o| \{\mathcal{S}, \mathbf{s}\}\right)\right]
\end{align}
to the set until we reach size $|\mathcal{S}|$. Algorithm~\ref{algo.:item_set_selection} summarizes the implementation of the active learning heuristic.

\section{Results}\label{sec:results}
\subsection{Theoretical Sample Complexity Bounds}\label{sec:theory}
We estimate the classifier at step $t$ using the unbiased estimator $\widetilde{\boldsymbol\theta}_t = \mathbb{E}[\boldsymbol\theta|\mathcal{F}_{t}]$. Using the arithmetic-geometric mean inequality, we bound this estimator \ac{MSE} as 
\begin{equation}\label{eq.MSE_sigma}
    \operatorname{MSE}_t = \operatorname{trace}(\Sigma_{\boldsymbol\theta|\mathcal{F}_{t}})\geq d|\Sigma_{\boldsymbol\theta|\mathcal{F}_{t}}|^{1/d},
\end{equation}
where $\Sigma_{\boldsymbol\theta|\mathcal{F}_{t}} = \mathbb{E}\left[(\boldsymbol\theta-\widetilde{\boldsymbol\theta}_t)(\boldsymbol\theta-\widetilde{\boldsymbol\theta}_t)^T|\mathcal{F}_t\right]$. Analogously to~\cite{Canal2019}, we observe that a necessary condition to obtain a low \ac{MSE} is for the determinant of the posterior covariance $|\Sigma_{\boldsymbol\theta|\mathcal{F}_{t}}|$ to be low. Next, we bound the expected number of iterations necessary to achieve a low enough posterior covariance. To facilitate our analysis, we first introduce a set of assumptions. These assumptions are not only useful for analytical tractability but also reflect common conditions or simplifications that align with real-world scenarios.

\begin{assumption} \label{assum.:xy_ind_H_given_theta}
    The annotator answer is independent of the history given the classifier, i.e., $p(o_{t} |\boldsymbol\theta, q_t, \mathcal{S}_t, \mathcal{F}_{t-1}) = p(o_{t} |\boldsymbol\theta, q_t, \mathcal{S}_t)$.
\end{assumption}

\begin{assumption} \label{assum.:y_ind_S}
    The label of an item is conditionally independent of the question and the rest of items in the query given the item, i.e., $p(y_t|\mathbf{x}_t, q_t, \mathcal{S}_t) = p(y_t|\mathbf{x}_t)$.
\end{assumption}
While Assumption \ref{assum.:y_ind_S} may not be perfectly accurate in every instance, this simplification is justified because the intrinsic score perceived by a human towards an item is primarily determined by the item itself. 

\begin{assumption} \label{assum.:minI}
    The information gain of the classifier given a human answer is lower bounded by a positive constant $L$. For ranking queries $I(\boldsymbol\theta; \mathbf{r}_t, l_t|\mathcal{F}_{t-1}) \geq L_r >0$, for exemplar selection queries $I(\boldsymbol\theta; i_t, y_t|\mathcal{F}_{t-1}) \geq L_s >0, \forall t$. 
\end{assumption}
This assumption implies that the query pool is sufficiently diverse and that the human feedback is strictly more informative than random noise.

\begin{assumption} \label{assum.:prior}
    The prior distribution of the classifier $p_0$ is uniform over a hypercube $[-M, M]^d$, for some $M>0.5$.
\end{assumption}
We assume a bounded parameter space to ensure realizability. Given this constraint, the uniform distribution is the least informative choice, as it is the unique maximum-entropy prior \cite{Cover2005ElementsTheory}.

Our main result is the following:
\begin{theorem} \label{thm.:Exact_T_bounds}
    Let $T_\epsilon=\min\{t: \left|\boldsymbol\Sigma_{\boldsymbol\theta|\mathcal{F}_t}\right|^{1/d} < \epsilon\}$ be the stopping time of Algorithm~\ref{algo.:ideal}. Under Assumptions \ref{assum_humanResponse} through \ref{assum.:prior}, $ \mathbb{E}[T_{\epsilon}]$ is bounded as
    \begin{equation*}
       \frac{d}{2}\frac{ \log_2 \frac{2M^2}{\pi e \epsilon}}{\log_2 N} \leq  \mathbb{E}[T_{\epsilon}] \leq \frac{d}{2L} \log_2  \frac{e^4d^2 M^2 }{2\sqrt{2}(d+2)\epsilon}  - 1,
    \end{equation*} 
    with $N = (|\mathcal{S}|+1)!$ and $L=L_r$ for ranking queries $q = q_{\text{rank}} $, and with $N = 2|\mathcal{S}| $ and $L=L_s$ for exemplar selection queries $q \in \{ q_{\text{high}}, q_{\text{low}} \}$.
\end{theorem}

\begin{proof}
    Let $h(\boldsymbol\theta; \mathcal{F}_{t}):= \mathbb{E}_{\boldsymbol\theta| \mathcal{F}_{t}}[-\log \mathbb{P}(\boldsymbol\theta| \mathcal{F}_{t})]$ be the entropy of the posterior distribution after observing $t$ interactions.
    For the lower bound, we note that Algorithm~\ref{algo.:ideal} selects the query deterministically as a function of its latest classifier estimator, i.e., $I(\boldsymbol{\theta}; q_t, \mathcal{S}_t| \mathcal{F}_{t-1}) = 0$, so that
    \begin{align}\label{eq:chain_rule}
        \mathbb{E}_{\mathcal{F}_{t}} \left[h(\boldsymbol\theta; \mathcal{F}_{t})\right]  &= h(\boldsymbol\theta; \mathcal{F}_{0}) - \sum_{j=1}^t I(\boldsymbol{\theta}; o_j, q_j, \mathcal{S}_j| \mathcal{F}_{j-1}) \nonumber\\
         &=  \begin{aligned}[t]
        d \log_2 2M - \sum_{j=1}^tI(\boldsymbol{\theta}; q_j, \mathcal{S}_j| \mathcal{F}_{j-1})& \\
        - \sum_{j=1}^tI(\boldsymbol{\theta}; o_j| q_j, \mathcal{S}_j, \mathcal{F}_{j-1})&
    \end{aligned} \nonumber \\
         &\geq 
          d \log_2 2M - t\log_2 N.
    \end{align}
    The first equality follows from the chain rule of mutual information, the second equality follows from Assumption \ref{assum.:prior}, and the last inequality holds because $I(\boldsymbol{\theta}; o_t| q_t, \mathcal{S}_t, \mathcal{F}_{t-1})$ is maximized for uniform outcomes.  
    Equation (\ref{eq:chain_rule}) implies 
    \begin{equation}\label{lowerbound_entropy}
        \mathbb{E}[T_\epsilon] \geq \frac{d \log_2 2M - \mathbb{E} \left[h(\boldsymbol{\theta}; \mathcal{F}_{T_\epsilon})\right] }{\log_2 N}.
    \end{equation}
    
    As Gaussian distributions maximize entropy for a given covariance \cite[Theorem 8.6.5]{Cover2005ElementsTheory}, so that 
    \begin{equation}\label{eq.lemma31_delta}
        h(\boldsymbol\theta; \mathcal{F}_{T_{\epsilon}}) \leq \frac{d}{2}\log 2\pi e|\Sigma_{\boldsymbol\theta|\mathcal{F}_{T_{\epsilon}}}|^{1/d} \leq \frac{d}{2}\log 2\pi e \epsilon.
    \end{equation}
    Combining Equation (\ref{lowerbound_entropy}) and (\ref{eq.lemma31_delta}), we obtain
    \begin{align*}
        \mathbb{E}[T_{\epsilon}] \geq \frac{d}{2}\frac{ \log_2 \frac{2M^2}{\pi e \epsilon}}{\log_2 N}.      
    \end{align*}

    To obtain the upperbound, we define the random variable $U_t := \frac{Z_t}{L} - t$, where $Z_t := -h(\boldsymbol\theta; \mathcal{F}_{t})$. From Lemma \ref{lemma:supermartingale} we know $U_t$ is a submartingale that fulfills the conditions of the optional stopping theorem \cite{williams1991probability}, so that
    \begin{multline}\label{eq.stopping_thm}
        \frac{\mathbb{E}[Z_{T_{\epsilon}}]}{L} - \mathbb{E}[T_{\epsilon}] \geq \frac{\mathbb{E}[Z_0]}{L} - \mathbb{E}[0]\\
        \iff \frac{\mathbb{E}[Z_{T_{\epsilon}}] - \mathbb{E}[Z_0]}{L}  \geq  \mathbb{E}[T_{\epsilon}].
    \end{multline} 
    By Assumption \ref{assum.:prior} the initial entropy is
    \begin{equation} \label{eq.Z_0}
        Z_0 = - h(\boldsymbol\theta; \mathcal{F}_{0}) = -d \log_2 2M.
    \end{equation}
    From Assumption \ref{assum.:minI} and Lemma \ref{lemma:LCC},
    \begin{equation}\label{eq.Z_minus_L}
        \mathbb{E}[Z_{T_\epsilon}] \leq \mathbb{E}[Z_{T_{\epsilon-1}}] - L.
    \end{equation}
    Since $\mathbb{P}(\boldsymbol\theta| \mathcal{F}_t)$ is log-concave, we invoke \cite[Lemma 3.1.]{Canal2019} 
    to obtain
        \begin{align} \label{eq.Z_T}
        Z_{T_{\epsilon -1}} &= -h(\boldsymbol\theta; \mathcal{F}_{T_{\epsilon -1}}) \leq  - \frac{d}{2}\log_2  \frac{2|\Sigma_{\boldsymbol\theta | \mathcal{F}_{T_{\epsilon -1}}}|^{1/d}}{e^4d^2 /(4\sqrt{2}(d+2))} \nonumber\\
        & \leq - \frac{d}{2}\log_2  \frac{8\sqrt{2}(d+2)\epsilon}{e^4d^2  }\nonumber\\
        & = \frac{d}{2}\log_2  \frac{e^4d^2  }{8\sqrt{2}(d+2)\epsilon}. 
    \end{align}
    Substituting Equations (\ref{eq.Z_0}), (\ref{eq.Z_minus_L}) and (\ref{eq.Z_T}) into Equation (\ref{eq.stopping_thm}), we obtain the desired upper bound of the expectation
    \begin{align*}
        \mathbb{E}[T_{\epsilon}] & \leq \frac{1}{L} \left( \mathbb{E}[Z_{T_{\epsilon}-1}] - L + d \log_2 2M\right)\\
        & \leq \frac{d}{2L} \left( \log_2  \frac{e^4d^2  }{8\sqrt{2}(d+2)\epsilon} + \log_2 4M^2\right) - 1 \\
        & \leq \frac{d}{2L} \log_2  \frac{e^4d^2 M^2 }{2\sqrt{2}(d+2)\epsilon}  - 1.
    \end{align*}
    
\end{proof}
    Theorem \ref{thm.:Exact_T_bounds} shows that the estimator uncertainty $\epsilon$ decays exponentially with the number of queries $\mathbb{E}[T_{\epsilon}]$. This result extends the upper bound in \cite{Canal2019} to non-equiprobable queries. The parameter $M$ appears because we allow for more freedom in the prior. The more constrained the prior is, i.e., the lower $M$, the fewer interactions we need. In contrast to the lower bound of \cite{Canal2019}, we generalize beyond binary queries; thus, the denominator $\log_2 N$ appears in the lower bound, suggesting a lower stopping time may be possible as the query complexity increases. We empirically corroborate the dependence of the stopping time on the question type and item set size in the next section. 
    The lower and upper bounds on the stopping time are monotonically increasing with $d$, which is consistent with the intuition that more interactions are required to learn classifiers in higher-dimensional embedding spaces. Crucially, both $N$ and $L$ are larger for ranking queries than for exemplar selection queries. Consequently, both bounds on the stopping time are lower for ranking queries, theoretically confirming that their higher information content necessitates fewer user interactions. 

\subsection{Empirical Sample Complexity Reduction}\label{sec:emp_sample_complexity}

\begin{figure*}[htp]
    \centering
    \includegraphics[width=0.67\linewidth]{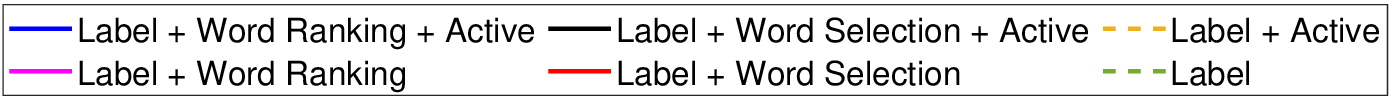}
    
    \begin{tabular}{cc}
        \includegraphics[width=0.36\linewidth]{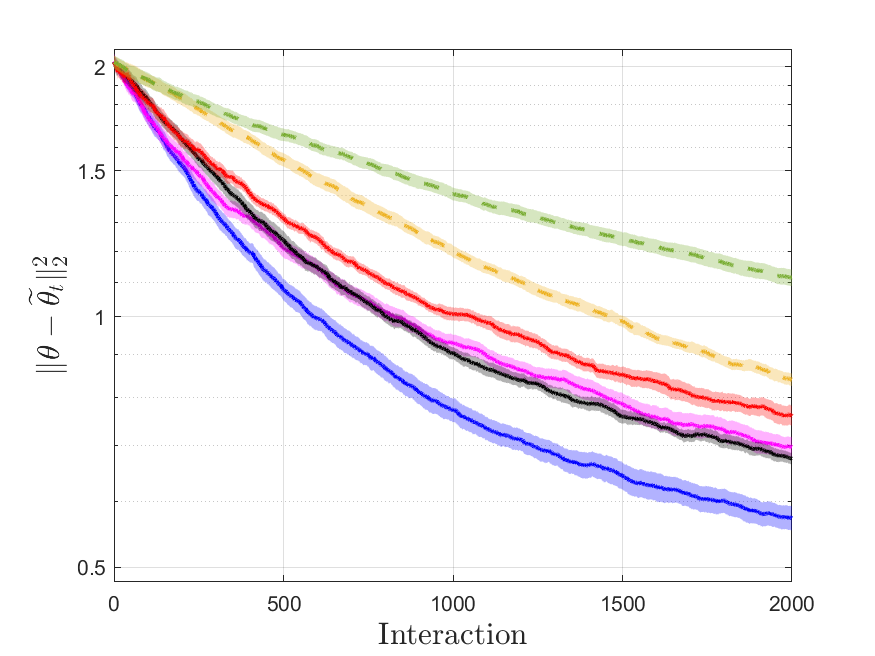} &
        \includegraphics[width=0.36\linewidth]{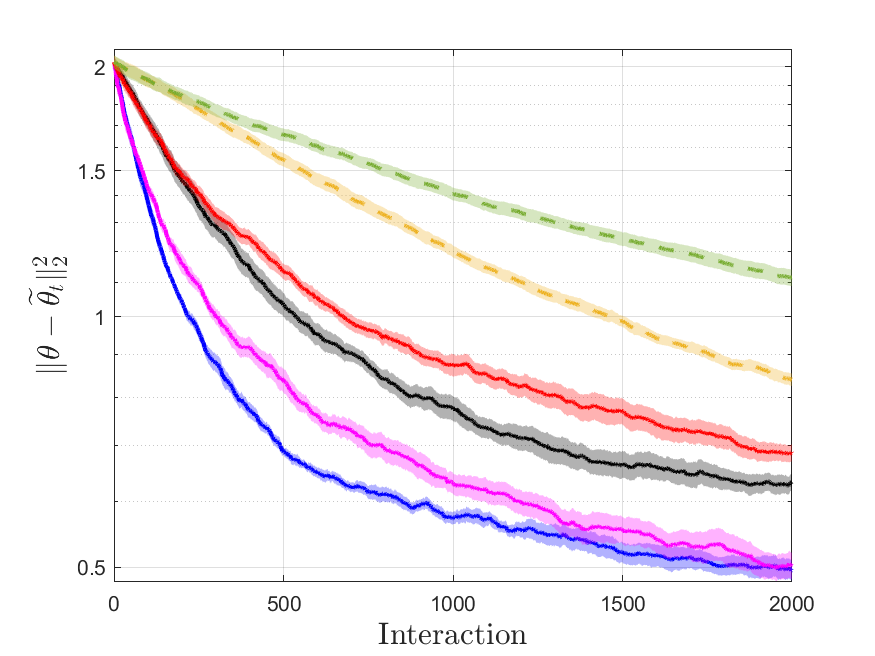} \\
        a) MSE vs. interaction for $|\mathcal{S}|=2$ & b) MSE vs. interaction for $|\mathcal{S}|=4$\\
        \includegraphics[width=0.36\linewidth]{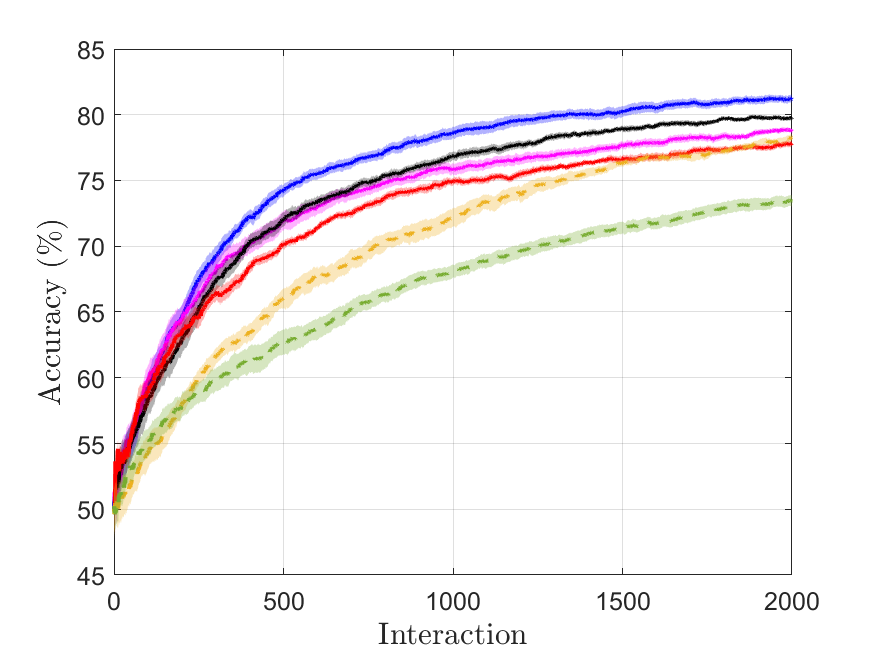} &
        \includegraphics[width=0.36\linewidth]{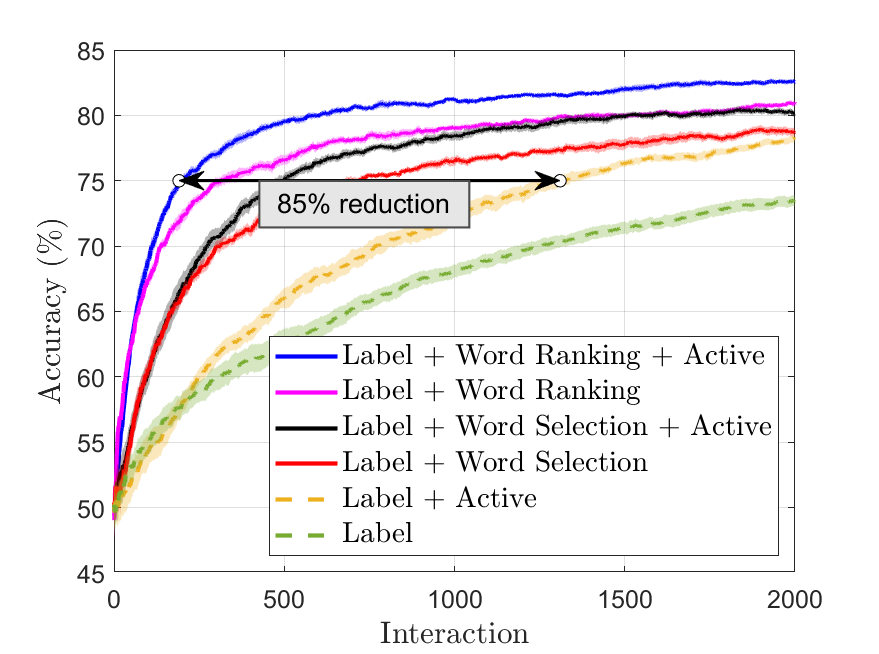}\\
        c) Accuracy vs. interaction for $|\mathcal{S}|=2$ & d) Accuracy vs. interaction for $|\mathcal{S}|=4$\\
     \end{tabular}
    \captionof{figure}{Performance of the human in the loop learning algorithms with human data on the word sentiment analysis task. All configurations are run with 10 different random initializations. The lines represent the mean of those experiments, while the shaded areas represent the standard error. Adding word selection or ranking to the queries together with actively selecting the word set reduces the number of iterations needed to achieve a good performance.}
    \label{fig.:wordSentimentResults}
\end{figure*}

\begin{figure}[htp]
    \centering
    \includegraphics[width=0.78\linewidth]{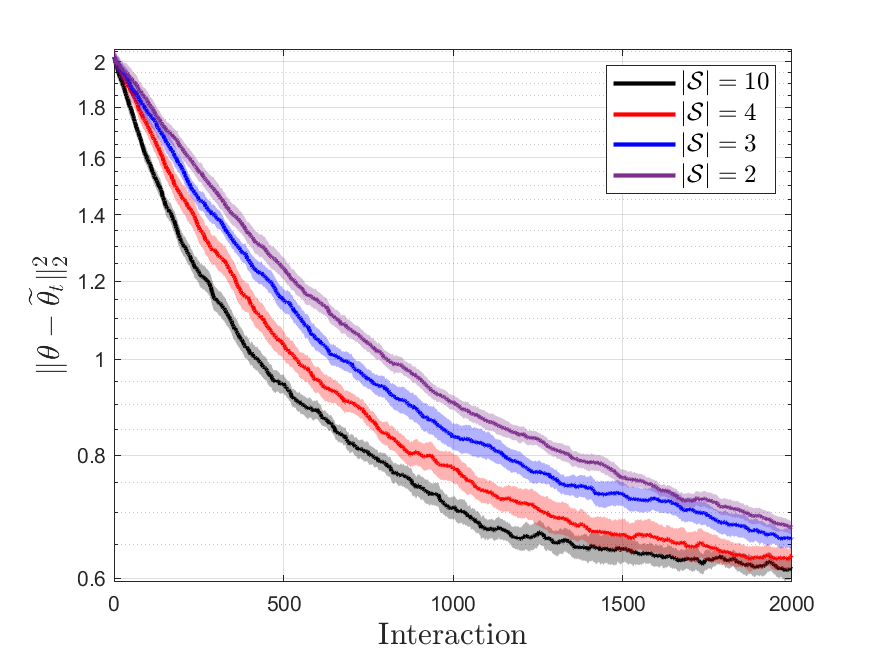}
    \caption{Performance of Algorithm~\ref{algo.:approx_selection} with human data on word sentiment classification. The larger the word set, the faster the decrease of MSE, as suggested by the lower bound in Theorem \ref{thm.:Exact_T_bounds}. }
    \label{fig.:MSE_vs_S}
\end{figure}

We empirically validate Algorithms \ref{algo.:approx_selection} and \ref{algo.:approx_ranking} on word and image classification tasks with existing crowdsourced datasets. To facilitate further exploration and validation by the research community, we provide the code for replicating our experiments\footnote{\url{https://github.com/BelenMU/HiTL-SentimentClassify/}} \cite{code}. 

\subsubsection{Word Sentiment Classification}
We first focus on the binary word sentiment classification task \cite{Nasukawa2003SentimentProcessing}. While humans can intuitively label words according to their connotation as positive (e.g., healthy) or negative (e.g., scary) \cite{Rana2018AApplications}, justifying the categorization in machine-interpretable terms proves challenging for most. We test the performance of our algorithms in learning this implicit knowledge from humans. We use the list of most frequent words in the decade of the 2000s
 \cite{Hamilton}. For every word $w$, we simulate the implicit human score by sampling from $\mathcal{N}({\mu}_{w}, \sigma_{w}^2)$, where ${\mu}_{w}$ and $ \sigma_{w}^2$ are the mean and variance of the valence score as given by the dataset \cite{Hamilton}. This simulation approximates a realistic distribution of human sentiment scores for each word and captures inter-subject variability in valence assessments.

We use existing 300-dimensional monolingual word embeddings \cite{word2vec} to which we prepend a 1; this way, the parameter $\boldsymbol\theta$ accounts for both the direction and the offset of the hyperplane characterizing the classifier. We define the ground-truth classifier $\boldsymbol{\theta} \in \mathbb{R}^{301}$ as the hyperplane that minimizes the labeling error over all words in the dataset. Figures~\ref{fig.:wordSentimentResults}a and \ref{fig.:wordSentimentResults}b show how the distance from the estimator to the ground truth decreases as more queries are collected. All the Bayesian strategies lead to an \ac{MSE} reduction, but the convergence speed markedly differs. The richer the query, the faster the decrease of the error: $q_{\text{rank}}$ outperforms $q_{\text{high}}$ and $q_{\text{low}}$, which in turn outperform traditional labeling queries. Additionally, actively choosing the items in the queries boosts performance. 
In fact, word selection between two actively chosen items surpasses ranking two randomly selected words. 
Algorithms~\ref{algo.:approx_selection} and \ref{algo.:approx_ranking} thus facilitate a faster reduction in \ac{MSE}, effectively reducing sample complexity.

The lower bound in Theorem~\ref{thm.:Exact_T_bounds} suggests that larger word sets should accelerate learning; we observe this behavior empirically when comparting Figures~\ref{fig.:wordSentimentResults}a and \ref{fig.:wordSentimentResults}b. Figure~\ref{fig.:MSE_vs_S} confirms this effect for word selection queries: after 1000 iterations, the \ac{MSE} is about 20\% lower when the annotator chooses from 10 options instead of 2. 

Beyond \ac{MSE}, we also assess the word classification accuracy. Given a word $w$ with embedding $\mathbf{x}_w$, we report the prediction $\operatorname{sign}(\widetilde{\boldsymbol{\theta}}^T\mathbf{x}_w)$ is accurate when it matches its label $\operatorname{sign}(\mu_w)$. To measure the predictor $\widetilde{\boldsymbol{\theta}}$ accuracy, we consider the words with $|P[Y_w = 1] - 0.5| = |\int_{0}^\infty f_{\mathcal{N}}(\mu_w, \sigma_w^2)-0.5| \leq 0.1$, where $f_{\mathcal{N}}$ represents the probability density function of a normal distribution. This range 
is selected to avoid words that have a completely neutral score, like ``branch'' or ``mouth.'' We focus on words with stronger sentiment scores, for which the prediction accuracy of our algorithm can be most meaningfully assessed.

Figures \ref{fig.:wordSentimentResults}c and \ref{fig.:wordSentimentResults}d show how classification accuracy evolves with the number of interactions.  As a baseline, using only randomly chosen labeling queries leads to a slow increase in accuracy, requiring more than 2000 interactions to reach 75\% accuracy. When labels are collected using active learning, the same accuracy is achieved after about 1300 labeling queries. Allowing the annotator to also select the most positive or most negative word among four candidates yields further gains: with randomly chosen word sets, 75\% accuracy is reached in roughly 700 interactions, dropping to about 500 interactions when the word sets are chosen actively. Ranking queries $q_{\text{rank}}$ offer the most substantial acceleration, requiring just $315$ rankings of four random words or $196$ rankings of actively chosen words to match the $75\%$ accuracy. This represents a notable reduction of approximately $85\%$ in the number of interactions needed compared to active labeling queries. These results demonstrate that our method not only improves estimator alignment metrics but also enhances the performance of the downstream classification task. Our results confirm the efficiency and practicality of Algorithms~ \ref{algo.:approx_selection} and \ref{algo.:approx_ranking} for valence classification. \looseness=-1

\subsubsection{Image Aesthetic Classification}
We further evaluate our approach on the task of binary image aesthetic classification. While humans naturally perceive the beauty of an image, automatically quantifying this aesthetic quality remains a challenging computational task \cite{anwar2021survey}. To test the efficiency of our algorithms in learning this subjective quality, we utilize the Aesthetic Visual Analysis (AVA) dataset \cite{Murray2012AVA:Analysis}. We focus on a subset of over 20,000 landscape images, which we embed into a 768 dimensional space with CLIP \cite{Radford2021LearningSupervision} and prepend a one. Each image is associated with a distribution of scores between 1 and 10 collected from an online photography community
. We define the ground-truth binary labels by using the dataset median mean score as a threshold $\tau = 5.58$, creating two balanced classes.  To simulate human evaluations, we leverage the crowdsourced score distributions: for every queried image, we sample a score from its empirical distribution and deterministically map these scores to query responses. Namely, we compare the sampled scores against the threshold $\tau$ to simulate human labeling; we select the image with the maximum or minimum sampled score to answer $q_{\text{high}}$ and $q_{\text{low}}$, respectively; and we order the images according to their sampled scores to simulate a human ranking. Figure~\ref{fig:results_AVA} shows how the classification accuracy of the learned classifier evolves as more feedback is gathered. Consistent with our previous experiments, richer queries require substantially fewer interactions to achieve a given accuracy. In particular, $q_{\text{rank}}$ yields the fastest accuracy gains, followed by the selection queries $q_{\text{high}}$ and $q_{\text{low}}$, while traditional labeling leads to the slowest improvement.  In fact, solely labeling does not reach a 65\% accuracy within 2000 interactions. In contrast, augmenting labels with image selection or ranking reduces the number of required interactions to approximately 400 and 900, respectively. Similarly, achieving a 63\% accuracy with $q_{\text{rank}}$ requires $84\%$ less interactions than traditional labeling. \looseness=-1

\begin{figure}
    \centering
    \includegraphics[width=0.78\linewidth]{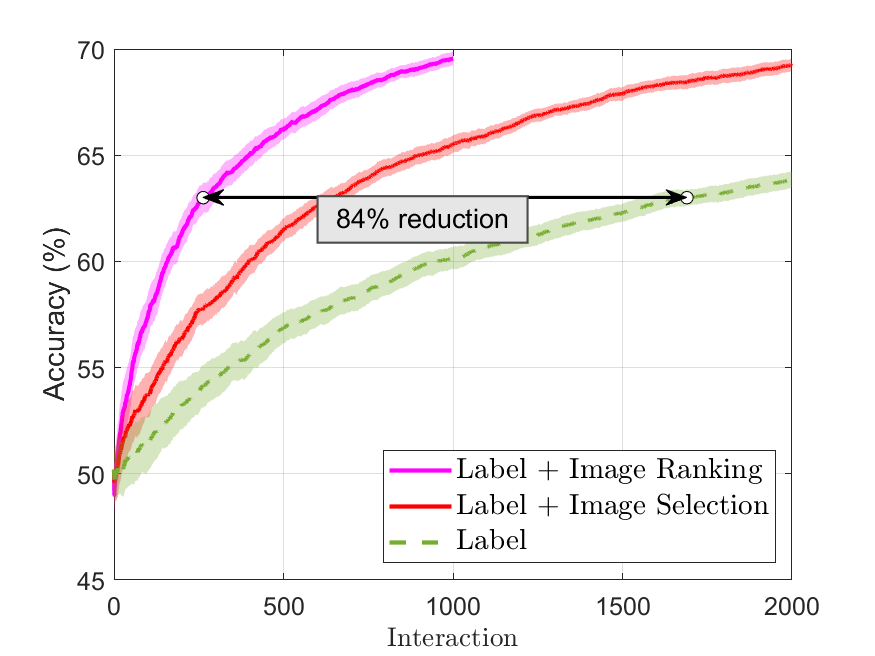}
    \caption{Performance of the human in the loop learning algorithms with human data on the image aesthetic classification task across 10 initializations. There are $|\mathcal{S}| = 4$ candidate images for ranking and selection questions. The accuracy increases faster when asking richer queries.   }
    \label{fig:results_AVA}
\end{figure}


\subsection{Empirical Time Savings in Cost Aware Query Selection}\label{sec:emp_time}
\subsubsection{Human Cost Model}\label{sec:human_cost_model}
To accurately maximize the information rate defined in Equation~\eqref{eq:rate}, we require a model of the human cost. Since the data collection costs are typically driven by collection times \cite{clancy2012active, sigurdsson2016much}, we define the cost as the expected time in seconds required for an annotator to answer a query. This response time depends on both the question type $q$ (ranking vs. selection) and the set size $|\mathcal{S}|$. 

We estimated this cost model in the context of the word sentiment classification task using crowdsourced experiments. We recruited participants through Prolific\footnote{The study was categorized as minimal risk research qualified for exemption status by the Institutional Review Board (IRB).} and recorded their response times for different combinations of question type and set size. Because we expect word selection questions $q_{\text{low}}$ and $q_{\text{high}}$ to incur similar times, we restricted data collection for selection queries to $q_{\text{high}}$. Each participant answered 20 $q_{\text{rank}}$ queries and 20 $q_{\text{high}}$ queries (in addition to 5 gold-standard questions per question type, used as attention checks and excluded from the analysis). We followed A/B testing guidelines and randomly assigned annotators to start with either the word selection or the word ranking queries. After excluding the participants who failed the attention checks, data from 83 participants remained for selection queries and 63 for ranking queries. Note that there is substantial overlap, most participants contributed to both question types. The graphic user interfaces used to record response times are shown in Figures \ref{fig:interface_selection} and \ref{fig:interface_ranking}. Additional details of the study may be found in Section~\ref{app:human_response_model} of the supplemental material. 

\begin{figure}
    \centering
    \includegraphics[width=0.82\linewidth]{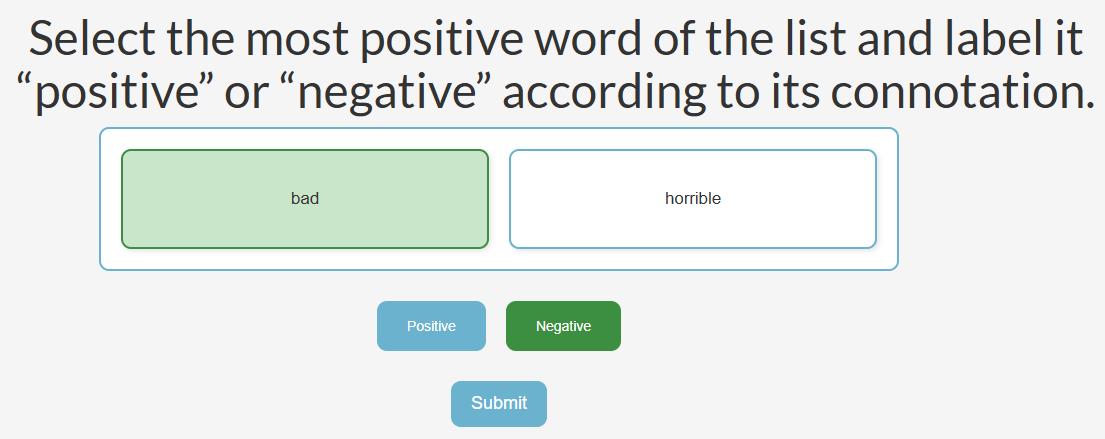}
    \caption{User interface for selection query. User must select the most positive word in the list, in this case ``bad'' and label the word according to its connotation, in this case ``negative''.}
    \label{fig:interface_selection}
\end{figure}
\begin{figure}
    \centering
    \includegraphics[width=0.82\linewidth]{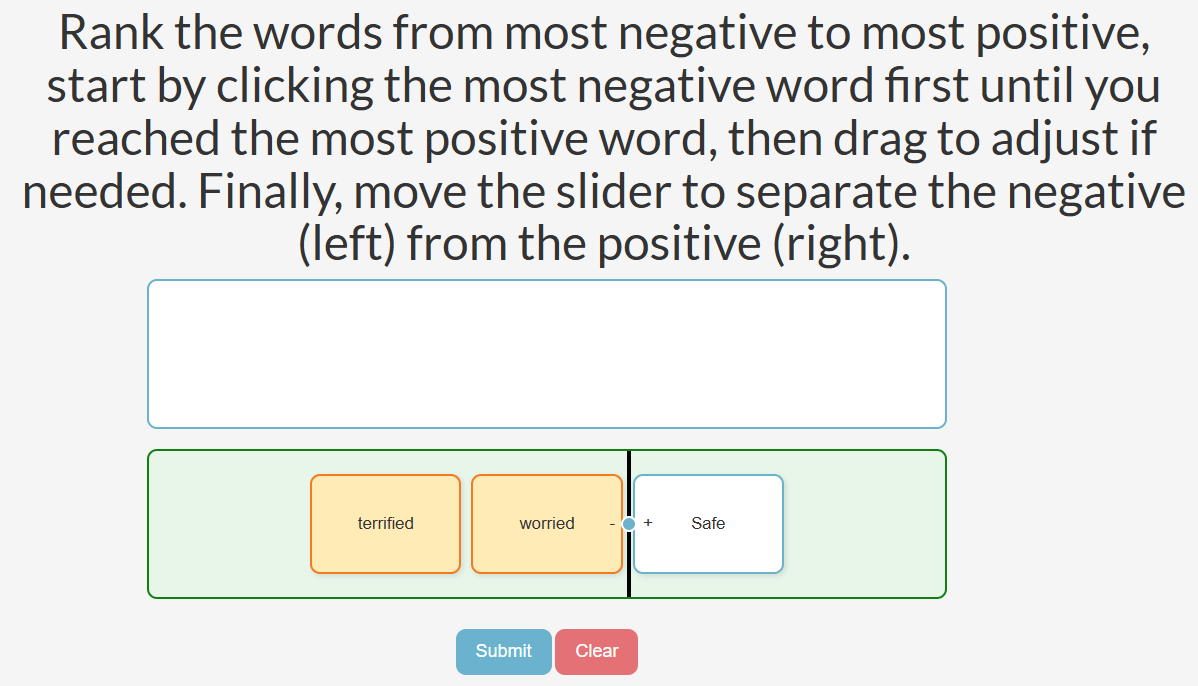}
    \caption{User interface for ranking query. The user must rank the words from the most negative word in the list ``terrified'' to the most positive ``safe'', and then use the vertical bar to separate words with positive and negative connotations.}
    \label{fig:interface_ranking}
\end{figure}

To select an appropriate parametric form for response time, we compared candidate models using a Vuong closeness test, a likelihood-ratio based test for comparing non-nested models. See Appendix~\ref{app:DataAnalysis} for further details on the candidate models and statistical results. Under our experimental conditions, a linear model $\widehat{t} = \beta_0 + \beta_1 |\mathcal{S}|$ describes the time response to selection queries significantly better than a logarithmic model $\widehat{t} = \beta_0 + \beta_1 \log|\mathcal{S}|$ ($p\text{-value} = 10^{-5}$). Similarly, the response times collected for ranking queries are significantly better explained by a linear model than by a purely quadratic model $\widehat{t} = \beta_0 + \beta_1 |\mathcal{S}|^2$ ($p\text{-value} = 4 \times 10^{-5}$). Motivated by these results, we fitted linear models with least-squares regression on individual response times and modeled the human response times as
\begin{align*}
    \widehat{t}_{\text{high}} = \widehat{t}_{\text{low}} &= 4.01 + 0.63 |\mathcal{S}| \text{  and  }
    \widehat{t}_{\text{rank}} = -0.32 + 4.41 |\mathcal{S}|.
\end{align*}
Figure~\ref{fig:linear_fit} shows the fitted predictions against the empirical mean response times for each set size. The linear models capture the overall increase in response time with set size. Although substantial trial‑to‑trial variability remains, this is typical in human response‑time data. As expected, the slope for ranking queries is steeper than for word selection queries, indicating that adding items to a set incurs a substantially higher time penalty when participants must produce a full ranking and labeling rather than identify and label a single most positive word from the set.

\begin{figure}
    \centering
    \vspace{-1pt}
    \includegraphics[width=0.78\linewidth]{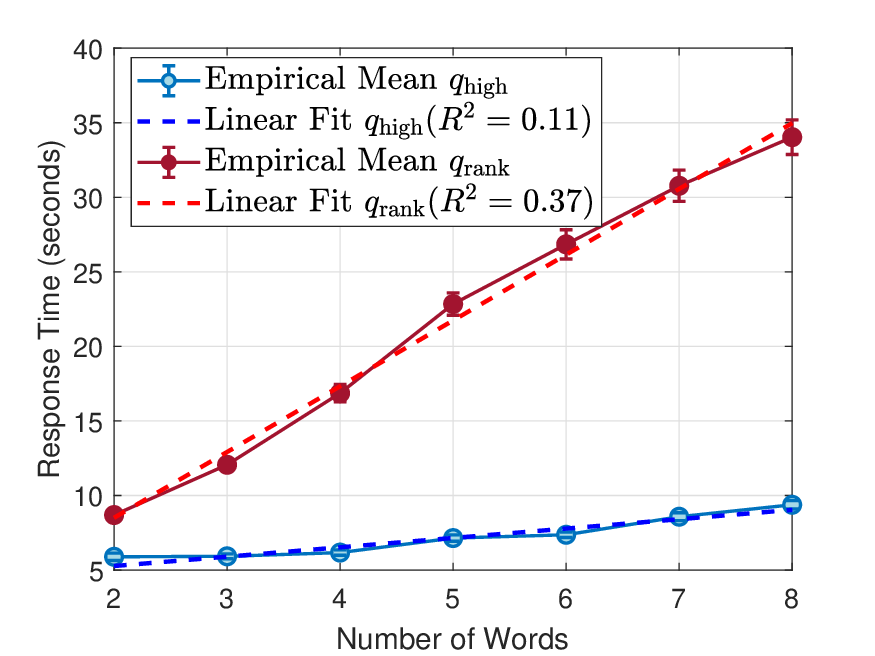}.
    \caption{Mean response time from crowdsourced experiments and linear model for different question types and word set sizes. The points show the empirical mean of the data, and the bars show the standard error among all the combined response times collected. We observe that the linear models closely track the empirical means across word set sizes.}
    \label{fig:linear_fit}
\end{figure}

    \begin{figure}
        \centering
        \includegraphics[width=0.81\linewidth, clip, trim=3cm 9cm 3cm 9cm]{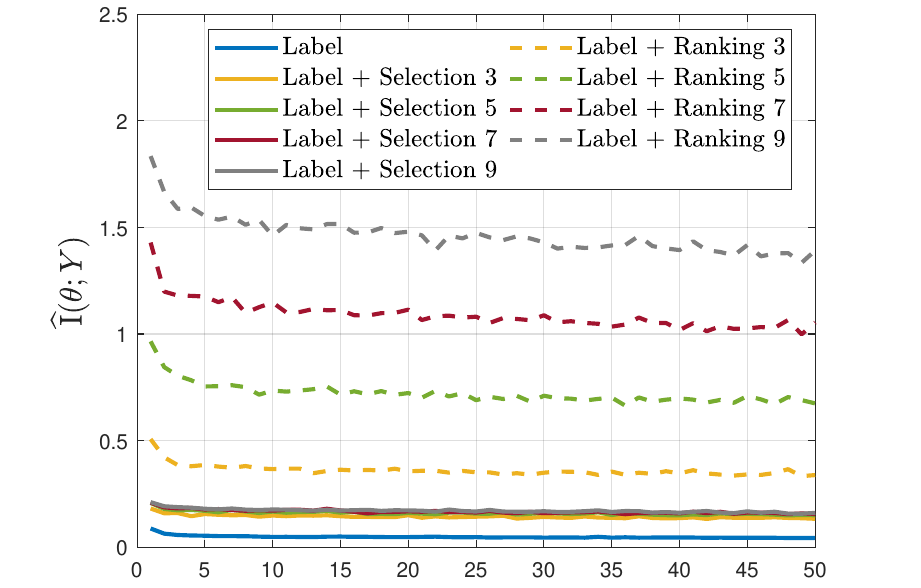}
        \caption{Estimated information gain for several query types. We empirically observe that the ratio of information gain between query types stays approximately constant with interactions.}
        \label{fig:Iratio_constant}
    \end{figure}

    \begin{figure}
        \centering
        \includegraphics[width=0.81\linewidth, clip, trim=3cm 9cm 3cm 9cm]{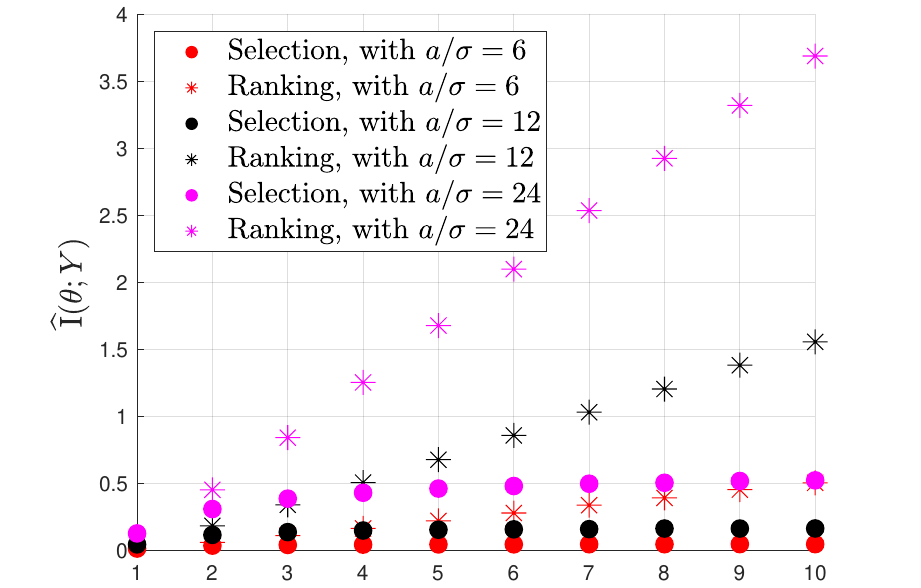}
        \caption{Estimated information gain for different slope and noise factor ratios~$a/\sigma$. }
        \label{fig:IG_NoiseFactor}
    \end{figure}
   \begin{figure*}
        \centering
        \hspace*{-0.35cm}\begin{tabular}{cc}
            \includegraphics[width=0.38\linewidth, clip, trim=3cm 9cm 3cm 9cm]{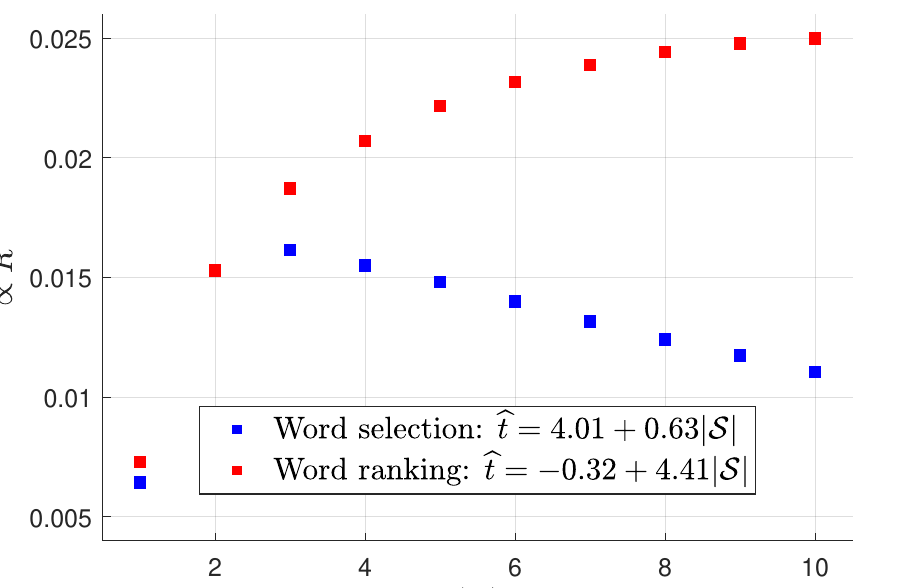} &
            \includegraphics[width=0.38\linewidth, clip, trim=3cm 9cm 3cm 9cm]{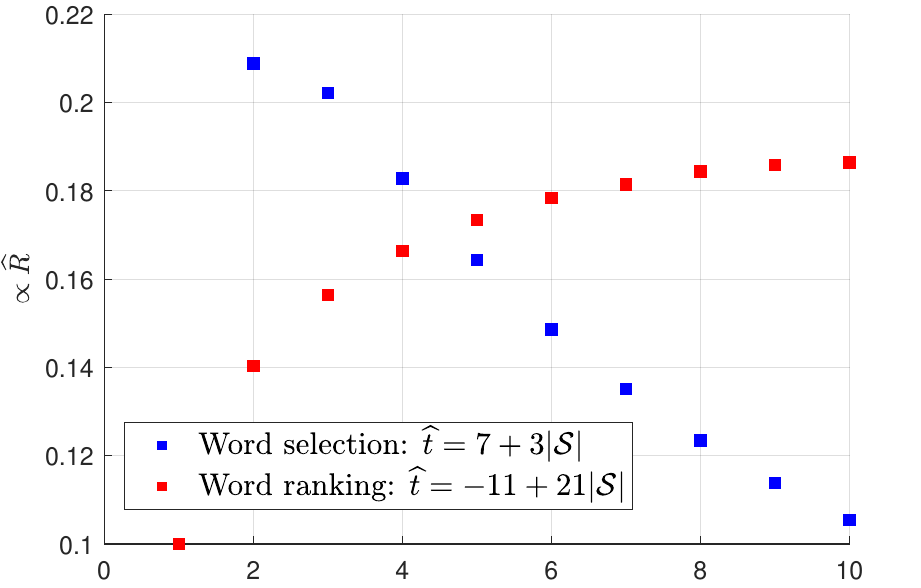} \\
            a) Response time modeled from experiments& b) Synthetic response time model\\
        \end{tabular}
        \caption{Predicted ratio of information gain over response time on the word sentiment classification task for different response time models. To maximize the rate with the interface  we tested, we should query ranking questions with 10 words. However, if an interface resulted on response time model in (b) we should query selection questions with 2 words. \looseness=-1} 
        \label{fig.:rate_vs_S}
        \vspace{\intextsep}
        \hspace*{-0.35cm}\begin{tabular}{cc}
            \includegraphics[width=0.4\linewidth, clip, trim=3cm 9cm 4cm 9cm]{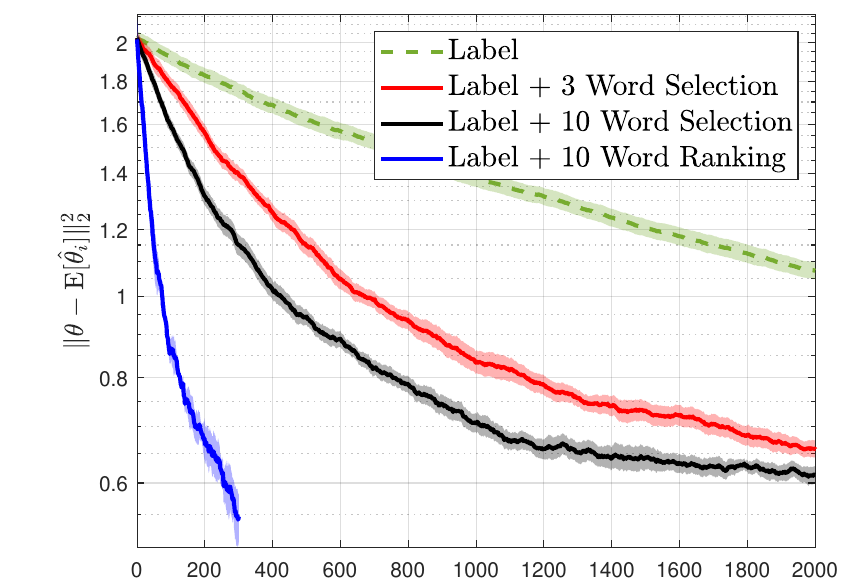} &
            \includegraphics[width=0.4\linewidth, clip, trim=3.5cm 9cm 3cm 9cm]{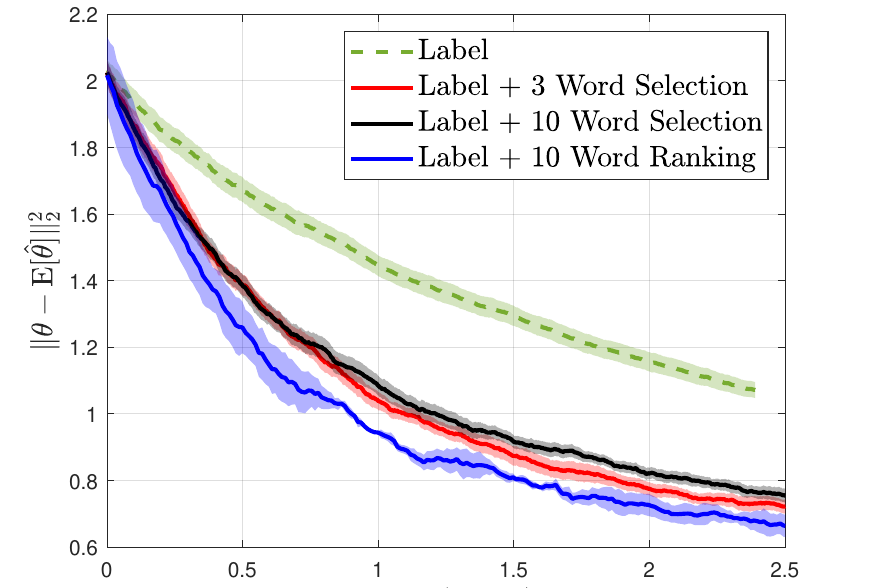} \\
            a) Sample complexity & b) Time cost \\
        \end{tabular}
        \caption{MSE evolution on the word sentiment classification task. The gains in sample complexity with complex queries is not so prominent in time cost. In fact, while the sample complexity decreases faster with word selections of size 10 vs. 3, the time cost decreases slower.}
        \label{fig.:iterations_vs_time}
    \end{figure*}  
    
\subsubsection{Information Gain Ratios}
Maximizing the information rate in Equation~\eqref{eq:rate} also requires estimates of the expected information gained for each combination of question type and set size. Directly recomputing these quantities at every iteration would be computationally expensive. Instead, we exploit an empirical regularity in the information gain ratios across queries. 

As illustrated in Figure~\ref{fig:Iratio_constant}, the ratios of information gain between different queries remain approximately constant across interactions. In particular, more complex queries consistently yield higher information, and in all our experiments, we observe that these relative advantages are stable not only across iterations but also across different initializations of the algorithm. This behavior allows us to estimate the ratios of expected information gain for each question type and set size relative to labeling from a small number of computations. Concretely, we compute the expected information gain as in Subsection~\ref{sec:AL_heuristic} for a few initial conditions and iterations. Then, we use the resulting average information gain ratio relative to the labeling query as a proxy for the proportional information gain in Equation~\eqref{eq:rate}. The values of these proxies strongly depend on the annotator noise~$\sigma$, as well as on the slope of the linear relationship between item distance to the classifier and score, i.e.,~$a$ in Equation~\eqref{eq.linear_score}. This dependence is shown in Figure~\ref{fig:IG_NoiseFactor}. When $a/\sigma$ is large, it is easier to distinguish the score ordering of items that are close in the embedding space, thereby increasing the benefit of richer queries. For a fixed set size~$\mathcal{S}$, ranking queries consistently have a higher estimated information gain than selection queries. For a fixed question type, increasing $|\mathcal{S}|$ raises the expected information gain. However, these improvements exhibit diminishing returns, especially when~$a/\sigma$ is small.

\subsubsection{Question and Set Size Selection}

Our goal is to select both the question type and the set size $|\mathcal{S}|$ that maximize information rate, by balancing information gain and human effort. Once the response time and the relative information gains of different query configurations have been estimated, we select the combination according to Equation~\eqref{eq:rate}. Figure~\ref{fig.:rate_vs_S} summarizes this trade-off for the word sentiment classification task. Figure~\ref{fig.:rate_vs_S}a shows the predicted information rate for each question and set size when using the response time models fitted from human experiments. For any fixed $|\mathcal{S}|$, ranking queries achieve a higher information rate than selection queries. The two question types exhibit different behavior as~$|\mathcal{S}|$ increases. For ranking, the information rate grows monotonically with~$|\mathcal{S}|$ over the feasible range ($|\mathcal{S}| \leq 10$). For selection questions, however, the information rate peaks at~$|\mathcal{S}| = 3$ and decreases for both smaller and larger set sizes.
This optimal query depends strongly on the underlying information gain estimation and cost models. For example,  Figure~\ref{fig.:rate_vs_S}b illustrates the information rates when the response time models are altered; in this case, the optimal query becomes a selection question with two words.

Different query selections translate into distinct learning behaviors. Figure~\ref{fig.:iterations_vs_time}a compares the reduction in \ac{MSE} as a function of the number of interactions on the word sentiment classification task. Larger and richer queries exhibit a strong advantage in sample complexity. 
When measured in wall-clock time, however, the picture changes. As shown in Figure~\ref{fig.:iterations_vs_time}b, the additional time required for larger sets can offset their sample-complexity gains, so that selection from 10 words reduces error more slowly in time than selection from 3 words. Consistent with our analysis, ranking queries with 10 items offer clear benefits in both sample complexity and time, but the improvements in time are less pronounced than in interaction count. These results underscore the importance of optimizing for information rate rather than sample complexity alone when costs are query dependent.

\FloatBarrier

\section{Conclusion and Future Work}
We have presented a framework for incorporating nuanced expert feedback into interactive learning. By exploiting embedding geometries, we have designed human-in-the-loop algorithms that use exemplar selection and ranking queries, providing richer supervision than standard label queries. We have shown, both theoretically and empirically, that these queries reduce sample complexity and accelerate learning. We have also proposed a query-selection strategy that accounts for query-dependent costs and demonstrated time savings on a word sentiment classification task, \emph{moving towards more cost-effective alignment between models and human expertise}. 

Several directions remain open. The relationship between scores and embeddings suggests new query types, such as asking which item a user can label most confidently or is most uncertain about. As in much of the literature, we assume human responses are conditionally independent given latent parameters, which may fail due to context effects or fatigue. Future work includes learning user-specific behavior models or adapting query policies online to user state.

\bibliographystyle{IEEEtran}
\bibliography{references, references-2}
\clearpage
\appendices
\section{Linear Relationship between Embeddings and Classifiers}\label{sec:app.LinearRelationship}
Assumption \ref{assum_humanResponse} is motivated by the observed relationship between implicit human scores and the geometry of existing embedding spaces. As Figure~\ref{fig.:score_vs_distance} shows, there exists a linear relationship between human scores and the distance from items to the \ac{MMSE} classifier in several standard embedding spaces.

\subsection{Empirical Evidence Across Domains}
Previous influential work \cite{Russell1977} identifies three fundamental dimensions of meaning: \textit{valence}, representing the spectrum between positivity and negativity; \textit{arousal}, representing the contrast between active and passive emotions; and \textit{dominance}, capturing the power dynamics from submissive to dominant. Figures~\ref{fig.:score_vs_distance}a,~\ref{fig.:score_vs_distance_app}a~and~\ref{fig.:score_vs_distance_app}b demonstrate that the scores across all three dimensions, as provided by the \ac{NRC} dataset \cite{Mohammad2018}, vary linearly with the distance to the \ac{MMSE} classifier. This linear behavior is consistent across distinct and independently gathered datasets \cite{Hamilton} of valence scores: Figure~\ref{fig.:score_vs_distance}b shows the learned human scores of the top-5000 most frequent non-stop words in the decade of the 2000s, the same dataset used for the word sentiment classification tasks in Sections \ref{sec:emp_sample_complexity} and \ref{sec:emp_time}. Figure~\ref{fig.:score_vs_distance_app}c examines the adjectives appearing more than 100 times in the data from the 2000s, we observe a similar linear behavior in their mean score in these distinct datasets.

We construct our word embedding representations through three steps. First, each word is mapped to a 300-dimensional vector using the standard \textit{word2vec} mapping\footnote{\url{http://ixa2.si.ehu.es/martetxe/vecmap/en.emb.txt.gz}} \cite{word2vec}. Second, the embedded words are normalized to unit norm. Finally, following standard practice for linear classification, we prepend a constant value of 1 to each normalized vector, yielding 301-dimensional feature vectors. This augmentation allows the classifier $\boldsymbol{\theta} \in \mathbb{R}^{301}$ to capture both the decision boundary direction and offset.

Next, we show that the linear relationship holds in the visual domain. Figure~\ref{fig.:score_vs_distance}c shows average aesthetic scores from human raters for 21,979 landscape images in the AVA dataset \cite{Murray2012AVA:Analysis} versus their distance to the MMSE classifier in the 769-dimensional (we prepend a 1) ViT-L/14 embedding space pretrained on CLIP \cite{Radford2021LearningSupervision}. The classifier separates high and low aesthetic images using the median score (5.48) as threshold. We observe a strong correlation between score and distance. 

Figure~\ref{fig.:score_vs_distance}d compares the ages of 23,625 aligned and cropped face images from the UTKFace dataset \cite{Zhang2017AgeAutoencoder} against their distance to a classifier separating faces under 30 years from older faces. Face images are embedded using InceptionResnetV1 pretrained on VGGFace2 \cite{Cao2018VGGFace2:Age}, yielding 512-dimensional vectors to which we prepend a constant value of 1. We observe a linear score-distance relationship again in this image embedding space.

    \begin{figure*}[t]
        \centering
        \begin{tabular}{c}
            \begin{tabular}{ccc}                
                \includegraphics[width=0.3\linewidth]{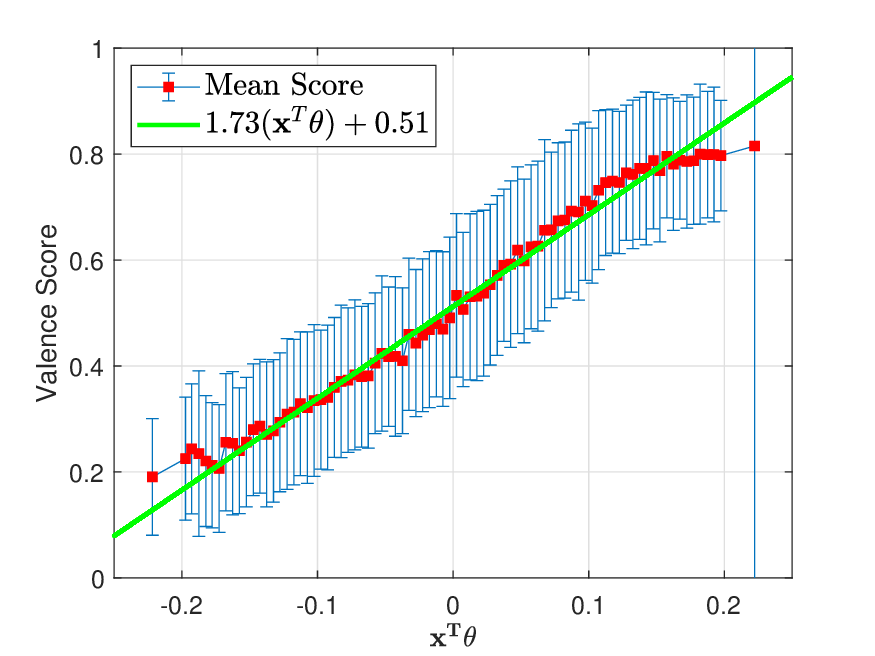} &
                \includegraphics[width=0.3\linewidth]{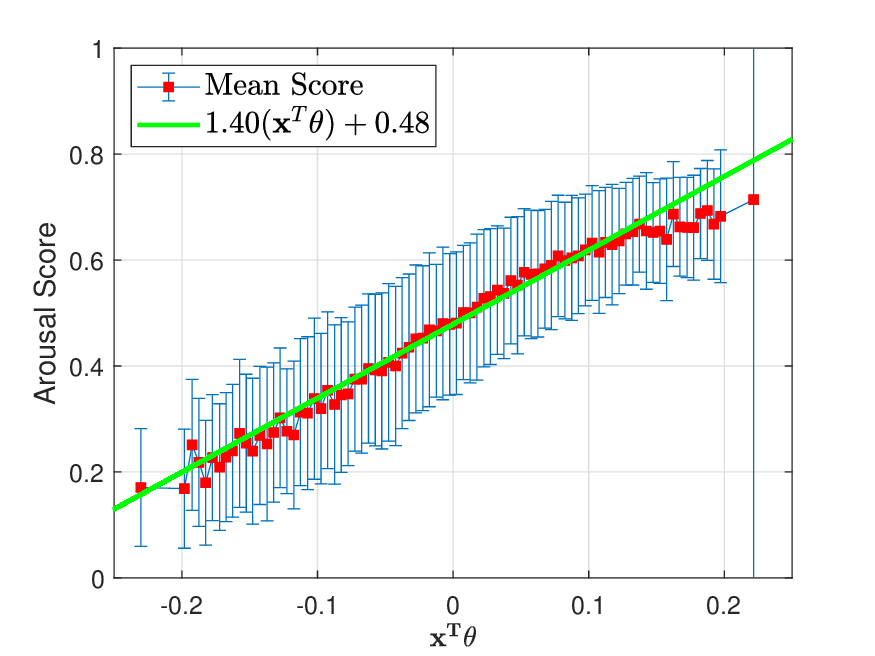} &\includegraphics[width=0.3\linewidth]{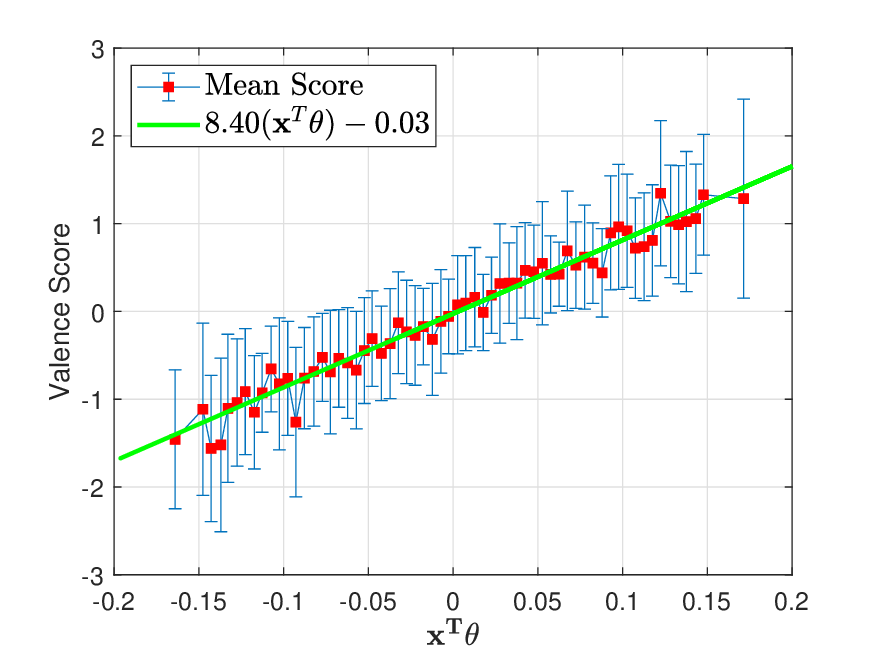}
                \\
                a) Valence in  NRC & b) Arousal in  NRC & Valence of adjectives in SocialSent 
            \end{tabular}
        \end{tabular}
        
        \caption{Relationship between the empirical mean score of words, given by the \ac{NRC} \cite{Mohammad2018} or SocialSent \cite{Hamilton} lexicons, and the distance of their word2vec embeddings to the ground truth classifier. We observe there is an approximately linear relationship.}
        \label{fig.:score_vs_distance_app}
    \end{figure*}

\subsection{Gumbel Noise Model Validation}\label{sec:app.Gumbel}
We model the noise in Equation~\eqref{eq.linear_score} as extreme-value distributed because it yields the Boltzmann choice model, a canonical framework in behavioral economics and discrete choice theory \cite{Georgii2011GibbsTransitions}. To validate this modeling assumption empirically, we examine the residuals $\delta_i = \text{score}(\mathbf{x}_i) - (a\mathbf{x}_i^T\boldsymbol\theta + b)$ against both  $\operatorname{Gumbel-Max}$ and $\operatorname{Gumbel-Min}$ distributions. 

Figure~\ref{fig:Gumbel_fit} shows diagnostic plots for the word valence (a) and image aesthetic (b) datasets. The $\operatorname{Gumbel-Max}$ distribution shows a superior fit, with Kolmogorov-Smirnov statistics of $0.095$ for the word valence dataset and $0.030$ for the image aesthetic dataset (indicating maximum CDF deviations of 3\% and 9.5\%, respectively). Visually, the residual histograms align well with the fitted PDFs, the empirical and theoretical CDFs nearly overlap, and the Q-Q plot points closely follow the theoretical line across the central range of the distribution. While some misspecification appears in the extreme tails, the $\operatorname{Gumbel-Max}$ distribution provides both tractable human decision-making models and a practical approximation for 60-80\% of the data distribution in both datasets.  

The $\operatorname{Gumbel-Min}$ distribution shows acceptable but slightly weaker fit, with Kolmogorov-Smirnov statistics of $0.136$ for word valence and $0.123$ for image aesthetics. While these values indicate larger maximum deviations compared to $\operatorname{Gumbel-Max}$, they remain within conventional bounds for acceptable model fit in behavioral data analysis. 
\begin{figure*}
    \centering
        \begin{tabular}{c}
        \includegraphics[width=0.8\linewidth]{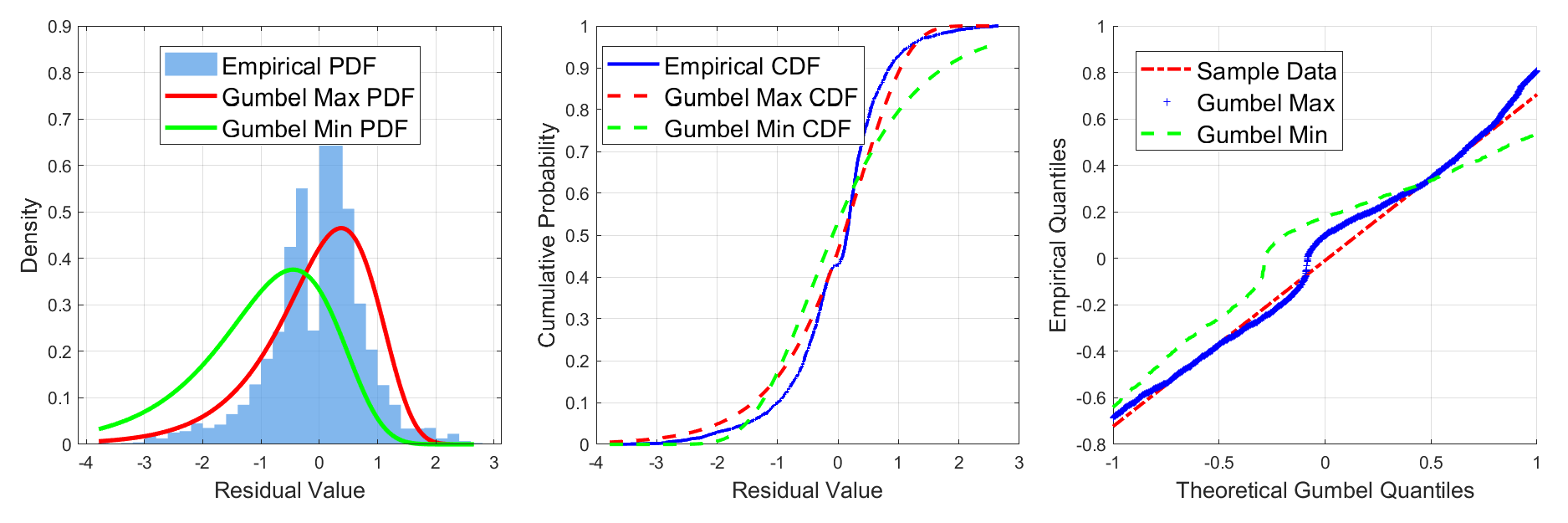}\\
        a) Word Valence\\
        \includegraphics[width=0.8\linewidth]{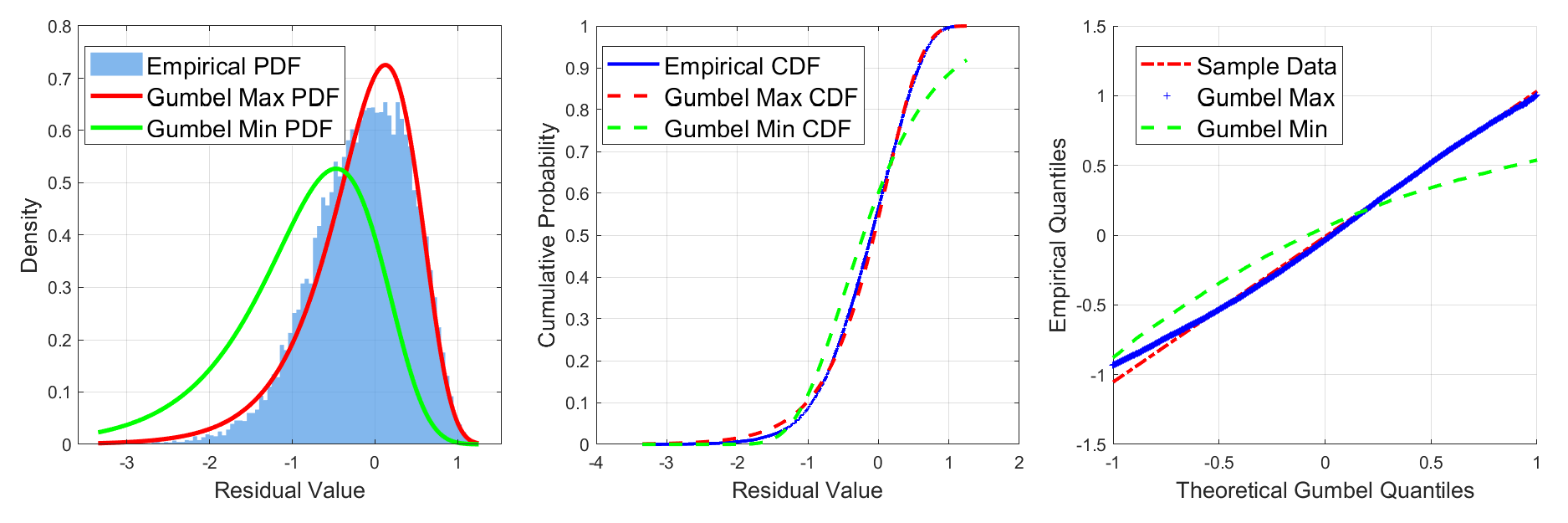}\\
        b) Image Aesthetic
        \end{tabular}
    \caption{Goodness-of-fit diagnostics for the $\operatorname{Gumbel-Max}$ and $\operatorname{Gumbel-Min}$ noise assumption in the linear score model. Each row shows, from left to right, the empirical residual density vs. fitted $\operatorname{Gumbel}$ PDFs, the empirical CDF vs. fitted $\operatorname{Gumbel}$ CDFs, and a Q–Q plot comparing empirical and theoretical $\operatorname{Gumbel}$ quantiles.}
    \label{fig:Gumbel_fit}
\end{figure*}

\section{Auxiliary Lemmas and Proofs}
 \begin{lemma}\label{lemma:LCC}
        The posterior distribution of the classifier given the history $p(\boldsymbol\theta | \mathcal{F}_t)$ is \ac{LCC}.
    \end{lemma}
    \begin{proof}
        \begin{align*}
            p_t(\boldsymbol\theta) :&= p(\boldsymbol\theta | \mathcal{F}_t)  = p(\boldsymbol\theta | o_t, q_t, \mathcal{S}_t, \mathcal{F}_{t-1}) \\
            &\stackrel{(1)}{=}\frac{p(o_t|q_t, \mathcal{S}_t, \boldsymbol\theta, \mathcal{F}_{t-1}) p(\boldsymbol\theta | q_t, \mathcal{S}_t,\mathcal{F}_{t-1})}{p(o_t |q_t, \mathcal{S}_t, \mathcal{F}_{t-1}) }\\
            & \stackrel{(2)}{=}\frac{p(o_t|q_t, \mathcal{S}_t,\boldsymbol\theta) }{p(o_t | \mathcal{F}_{t-1}) }p(\boldsymbol\theta | \mathcal{F}_{t-1})\\
            & \stackrel{(3)}{=} \prod_{j=1}^t \frac{p(o_j|q_j, \mathcal{S}_j,\boldsymbol\theta) }{p(o_j |\mathcal{F}_{j-1}) }p_0(\boldsymbol\theta)\\
            & \stackrel{(4)}{=}  \frac{\prod_{j=1}^tp(o_j|q_j, \mathcal{S}_j,\boldsymbol\theta) }{p( \mathcal{F}_{t}) }p_0(\boldsymbol\theta)
        \end{align*}
        where 
        \begin{enumerate}[label=(\arabic*)]
            \item follows from Bayes theorem,
            \item follows from Assumption \ref{assum.:xy_ind_H_given_theta}, and because given the past history the query and word set are selected deterministically by the algorithm,
            \item follows by induction,
            \item follows from the law of total probability.
        \end{enumerate}
        Combining Equations \eqref{eq.likelihood_label}
        -\eqref{eq.likelihood_qrank} and Assumption \ref{assum.:y_ind_S}, we deduce that the likelihood of the human response,  $p(o_j|q_j, \mathcal{S}_j,\boldsymbol\theta)$ , is given by a product of softmax and logistic functions, both of which are \ac{LCC}. From Assumption \ref{assum.:prior}, the prior  $p_0(\boldsymbol\theta)$  is sampled from a uniform distribution, which is also \ac{LCC}. The product of \ac{LCC} functions is also \ac{LCC}, thus the posterior $p_t(\boldsymbol\theta) $ is \ac{LCC}. 
    \end{proof}

    \begin{lemma} \label{lemma:likelihood}
        The likelihood of any word and label pair $\mathbf{x} \in \mathcal{X}$ and $y \in \{-1, 1\}$ is lower bounded as
        \begin{equation*}
            \min_{\mathbf{x} , y, \mathcal{S}, \boldsymbol\theta, q \in \{q_{\text{high}}, q_{\text{low}}\}} \mathbb{P}\left[\mathbf{x} ,y |  \mathcal{S}, \boldsymbol\theta , q\right]  
            \geq {\gamma_1 \gamma_2}>0, 
        \end{equation*}
            with $\gamma_1  := \frac{\exp{\left(-\left|\frac{a}{\sigma}\right|M\sqrt{P_{\mathbf{x}}d}\right)}}{\exp{\left(-\left|\frac{a}{\sigma}\right|M\sqrt{P_{\mathbf{x}}d}\right)}+ (|\mathcal{S}|-1)\exp{\left(\left|\frac{a}{\sigma}\right|M\sqrt{P_{\mathbf{x}}d}\right)}} $ and $\gamma_2 = \frac{1}{1 +\exp{(wM\sqrt{P_{\mathbf{x}}d})}}$.
    \end{lemma}
    \begin{proof}
        By construction, all word embeddings are bounded, i.e., $\|\mathbf{x}\|_2^2 \leq P_{\mathbf{x}}\}$. From Assumption \ref{assum.:prior}, we know the norm of the classifier is bounded by $\|\boldsymbol\theta\|_2^2 \leq d M^2$. We leverage these constraints to bound the likelihood. 

    Note that 
    \begin{equation*}
        \exp{\left(-\left|\frac{a}{\sigma}\right|M\sqrt{P_{\mathbf{x}}d}\right)}
        \leq 
        \exp{(k\mathbf{x}^T\boldsymbol\theta)}
        \leq \exp{\left(\left|\frac{a}{\sigma}\right|M\sqrt{P_{\mathbf{x}}d}\right)},
    \end{equation*}
    where $k = \frac{a}{\sigma}$ when $q = q_{\text{high}}$ and $k = -\frac{a}{\sigma}$ when $q = q_{\text{low}}$. Therefore, 
    \begin{align*}
        &\min_{\mathbf{x} , \mathcal{S}, \boldsymbol\theta, q} \mathbb{P}\left[\mathbf{x} | \mathcal{S}, \boldsymbol\theta , q\right] \\ & \geq \frac{\exp{\left(-\left|\frac{a}{\sigma}\right|M\sqrt{P_{\mathbf{x}}d}\right)}}{\exp{\left(-\left|\frac{a}{\sigma}\right|M\sqrt{P_{\mathbf{x}}d}\right)}+ (|\mathcal{S}|-1)\exp{\left(\left|\frac{a}{\sigma}\right|M\sqrt{P_{\mathbf{x}}d}\right)}} =:\gamma_1.
    \end{align*}

    In a similar fashion, we use Assumption \ref{assum.:y_ind_S} to bound
    
    \begin{align*}
    \min_{y, \mathbf{x}, \mathcal{S}, \boldsymbol\theta , q} & \mathbb{P}\left[y|\mathbf{x}, \boldsymbol\theta, q, \mathcal{S} \right] \\&= \min_{\mathbf{x}, \boldsymbol\theta} \min_{y \in \{-1, 1\}}\left(\frac{1}{1 + \exp{\left(y w (\boldsymbol\theta^T \mathbf{x})\right)}}\right)\\& \geq  \frac{1}{1 +\exp{(wM\sqrt{P_{\mathbf{x}}d})}} =: \gamma_2.
    \end{align*}

    We conclude the proof by combining the bounds    
    \begin{align*}
    &\min_{\mathbf{x} , \mathcal{S}, y,\boldsymbol\theta, q}  \mathbb{P}\left[\mathbf{x} ,y |  \mathcal{S}, \boldsymbol\theta , q\right] \\
    & \ \ =  \min_{\mathbf{x} , \mathcal{S}, y,\boldsymbol\theta, q} \mathbb{P}\left[\mathbf{x} |  \mathcal{S}, \boldsymbol\theta , q\right] \mathbb{P}\left[y |  \mathbf{x}, \mathcal{S}, \boldsymbol\theta , q\right] \geq \gamma_1\gamma_2.
    \end{align*}

    \end{proof}

\begin{corollary} \label{corollary:likelihood}
    The likelihood of any answer to a query in $q \in \mathcal{Q} = \{q_{\text{high}}, q_{\text{low}}, q_{\text{rank}}\}$ is lower bounded as
    \begin{equation*}
        \min_{o, \mathcal{S}, \boldsymbol\theta, q} \mathbb{P}\left[o |  \mathcal{S}, \boldsymbol\theta , q\right]  
        \geq 2^{\gamma_L}>0, 
    \end{equation*}
    with $\gamma_L = \log_2 \left(\gamma_1^{|\mathcal{S}|-1}\gamma_2^{|\mathcal{S}|} \right).$
    
\end{corollary}
\begin{proof}
    The likelihood of any answer to a ranking query, i.e., ranking order and threshold pair $(\mathbf{r}, l) $, is lower bounded as
    \begin{align*}
        & \min_{\mathbf{r}, l, \mathcal{S}, \boldsymbol\theta, q}  \mathbb{P}\left[\mathbf{r}, l |  \mathcal{S}, \boldsymbol\theta , q_{\text{rank}}\right]  \\
         & = \min_{\mathbf{r}, l, \mathcal{S}, \boldsymbol\theta, q}  \mathbb{P}\left[\mathbf{r} |  \mathcal{S}, \boldsymbol\theta , q_{\text{rank}}\right] \mathbb{P}\left[l | \mathbf{r}, \mathcal{S}, \boldsymbol\theta , q_{\text{rank}}\right] \\
        &= \min_{\mathbf{r}, l, \mathcal{S}, \boldsymbol\theta, q}  \prod_{j=1}^{|\mathcal{S}|-1} \mathbb{P}\left[r_j | \mathcal{R}_j, \boldsymbol{\theta}, q_{\text{high}}\right]  \times
        \\ & \quad \mathbb{P}\left[ y(\mathbf{x}_{r_1}, ... .r_{l-1}) = 1 \cap y(\mathbf{x}_{r_l}, ...,r_{|\mathcal{S}|}) = -1 | \mathbf{r}, \mathcal{S}, \boldsymbol\theta , q_{\text{rank}}\right]        
        \\  &\geq\gamma_1^{|\mathcal{S}|-1}\gamma_2^{|\mathcal{S}|}>0, 
    \end{align*}
    which is strictly lower than the lower bound from Lemma \ref{lemma:likelihood}, because $\gamma_1\gamma_2<1$.
\end{proof}

\begin{lemma} \label{lemma:entropy_dif}
The expected difference in entropy from one iteration to the next is bounded as
    \begin{equation*}
        \mathbb{E}_{o_i}\left[|h(\boldsymbol\theta; \mathcal{F}_i)-h(\boldsymbol\theta; \mathcal{F}_{i-1})|\right] \leq \gamma< \infty,
    \end{equation*}
    with $\gamma = 16d +d \log_2 2\pi e d - 2\gamma_L$.
\end{lemma}
\begin{proof}
    We extend \cite[Lemma A.2.]{Canal2019} to bound the expected posterior entropy difference for non-equiprobable and non-binary query schemes. We rewrite
    \begin{align}\label{eq.:entropy_difference_summands}
        |h(\boldsymbol\theta; \mathcal{F}_i)-h(\boldsymbol\theta; \mathcal{F}_{i-1})| = & \left(h(\boldsymbol\theta; \mathcal{F}_{i-1})-h(\boldsymbol\theta; \mathcal{F}_{i})\right)^{+} \\
        & + \left(h(\boldsymbol\theta; \mathcal{F}_i)-h(\boldsymbol\theta; \mathcal{F}_{i-1})\right)^{+}, \nonumber
    \end{align}
    where $x^+ := \max(x, 0)$ denotes the positive part. 
    
    Note that
    \begin{align*}
        -h(\boldsymbol\theta; \mathcal{F}_i)&  \stackrel{(1)}{\leq} \log_2 \mathbb{E}_{\boldsymbol\theta| \mathcal{F}_i}[p(\boldsymbol\theta|\mathcal{F}_i)] \\
        &  \stackrel{(2)}{=} \log_2 \mathbb{E}_{\boldsymbol\theta| \mathcal{F}_i} \left[ \frac{p(o_i|\boldsymbol\theta,q_i, \mathcal{S}_i, \mathcal{F}_{i-1}) }{p(o_i |q_i, \mathcal{S}_i, \mathcal{F}_{i-1}) }p(\boldsymbol\theta | \mathcal{F}_{i-1})\right]  \\
        & \stackrel{(3)}{\leq}\log_2 \mathbb{E}_{\boldsymbol\theta| \mathcal{F}_i} \left[\frac{p(\boldsymbol\theta | \mathcal{F}_{i-1}) }{p(o_i |q_i, \mathcal{S}_i, \mathcal{F}_{i-1}) } \right]\\
        & \stackrel{(4)}{\leq} -\frac{1}{2}\log_2 \left|\Sigma_{\boldsymbol\theta|\mathcal{F}_i}\right|+8d +\frac{d}{2} \log_2 d \\
        & \quad - \log_2 p(o_i |q_i, \mathcal{S}_i, \mathcal{F}_{i-1})\\
        &\stackrel{(5)}{=} -\frac{1}{2}\log_2 \left((2\pi e)^d\left|\Sigma_{\boldsymbol\theta|\mathcal{F}_i}\right|\right) + \frac{1}{2} \log_2 (2\pi e)^d\\
        & \quad +8d +\frac{d}{2} \log_2 d - \log_2 p(o_i |q_i, \mathcal{S}_i, \mathcal{F}_{i-1})\\
        &\stackrel{(6)}{\leq} -h(\boldsymbol\theta; \mathcal{F}_{i-1}) +8d +\frac{d}{2} \log_2 2\pi e d - \gamma_L\\
        &  := -h(\boldsymbol\theta; \mathcal{F}_{i-1}) + \frac{1}{2}\gamma.
    \end{align*}
    Thus, we bound the first summand in \eqref{eq.:entropy_difference_summands} as $\left(h(\boldsymbol\theta; \mathcal{F}_{i-1}) - h(\boldsymbol\theta; \mathcal{F}_{i})\right)^+ \leq  \frac{1}{2}\gamma $. The inequalities follow
    \begin{enumerate}[label=(\arabic*)]
        \item from Jensen's inequality,
        \item from Bayes Theorem and because the query is determined by the history,
        \item because $o_i$ is discrete and its likelihood is a valid probability distribution, thus the likelihood is at most 1, and the logarithm is monotonically increasing,
        \item applying the bound in \cite[Theorem 5.14]{Lovasz2007TheAlgorithms} to the \ac{LCC} isotropic $V = \Sigma_{\boldsymbol\theta|\mathcal{F}_i}^{-1/2}W$, where $W \sim p(\boldsymbol\theta|\mathcal{F}_i)$, together with the density of a linear transformation of a random variable.
        \item adding and subtracting $\frac{1}{2} \log_2 (2\pi e)^d$,
        \item from the maximum entropy distribution \cite[Theorem 8.6.5.]{Cover2005ElementsTheory} and Corollary \ref{corollary:likelihood}.
    \end{enumerate}

    To bound the second summand in \eqref{eq.:entropy_difference_summands}, we recall the non-negativity of mutual information
    \begin{align*}
        0 &\leq \mathrm{I}\left(\boldsymbol{\theta} ; o_i, q_i, \mathcal{S}_i | \mathcal{F}_{i-1} \right)  = \mathbb{E}_{o_i}\left[h(\boldsymbol\theta; \mathcal{F}_{i-1}) - h(\boldsymbol\theta; \mathcal{F}_i) \right] \\
        & = \mathbb{E}_{o_i}\left[\left(h(\boldsymbol\theta; \mathcal{F}_{i-1}) - h(\boldsymbol\theta; \mathcal{F}_i)\right) ^+ - \left(h(\boldsymbol\theta; \mathcal{F}_{i}) - h(\boldsymbol\theta; \mathcal{F}_{i-1})\right) ^+ \right]\\
        & \leq \frac{1}{2} \gamma - \mathbb{E}_{o_i}\left[\left(h(\boldsymbol\theta; \mathcal{F}_{i}) - h(\boldsymbol\theta; \mathcal{F}_{i-1})\right) ^+ \right].
    \end{align*}
    Therefore, $\mathbb{E}_{o_i}\left[\left(h(\boldsymbol\theta; \mathcal{F}_{i}) - h(\boldsymbol\theta; \mathcal{F}_{i-1})\right) ^+ \right] \leq \frac{1}{2} \gamma $. Combining the bounds we conclude
    \begin{align*}
        \mathbb{E}_{o_i} \left[\left|h(\boldsymbol\theta; \mathcal{F}_{i-1}) - h(\boldsymbol\theta; \mathcal{F}_i) \right| \right] \leq  \frac{1}{2} \gamma+\frac{1}{2} \gamma=\gamma.
    \end{align*}   
\end{proof}

\begin{lemma} \label{lemma:supermartingale}
    The random variable  $U_i := \frac{-h(\boldsymbol\theta ; \mathcal{F}_{i})}{L} - i$ is a submartingale that fullfils the conditions of the optional stopping theorem.
\end{lemma}
\begin{proof}
    The expectation of $U_i$ given the previous values of the sequence is
    \begin{align}\label{eq.submartin2}
        \mathbb{E}[U_i | U^{i-1}] & = \frac{\mathbb{E}[Z_i|Z^{i-1}]}{L} - i \geq \frac{\mathbb{E}[Z_{i-1}|Z^{i-1}] + L}{L} - i\nonumber\\
        & =  \frac{Z_{i-1}}{L} - (i-1) = U_{i-1},
    \end{align}
    where $Z_i = -h(\boldsymbol\theta ; \mathcal{F}_{i})$. The first inequality follows from Assumption \ref{assum.:minI}. 
    
    We bound the expected increment per step with Lemma \ref{lemma:entropy_dif},
    \begin{align} \label{eq.U_dif}
        \mathbb{E}\left[|U_{i+1} - U_i| \right]& = \mathbb{E}\left[\left| \frac{Z_{i+1}}{L} - i -1 - \frac{Z_i}{L} + i\right|\right] \nonumber \\ & = \frac{ \mathbb{E}\left[\left| Z_{i+1}-Z_i\right|\right]}{L} + 1\leq \frac{\gamma}{L} + 1.
    \end{align}

    Lastly, we want to show that $\mathbb{E}[T] <   \infty$. Note that a bound on the determinant of the posterior covariance implies a bound on the posterior entropy, so we introduce a threshold $0< \tau < \infty$ dependent on $\epsilon$ such that the stopping time becomes $T:= \min \{i : -h(\boldsymbol{\theta};\mathcal{F}_i) > \tau\}= \min \{i : U_i > \frac{\tau}{L} - i\}$. 

    From \cite[Theorem 4.2.9.]{Durrett2019Probability:Examples} we know that $-h(\boldsymbol{\theta};\mathcal{F}_{i\wedge T})$ is also submartingale, where $i\wedge T:= \min \{i, T\}$. Additionally, it is bounded by
    \[
    -h(\boldsymbol{\theta};\mathcal{F}_{i\wedge T}) \leq \tau + \frac{\gamma}{L} + 1,
    \]    
    by definition of $T$ and \eqref{eq.U_dif}. Therefore, by the Martingale convergence theorem \cite[Theorem 4.2.11.]{Durrett2019Probability:Examples}, as $i \rightarrow \infty$, $-h(\boldsymbol{\theta};\mathcal{F}_{i\wedge T})$ converges a.s. to a limit $H$ with $\mathbb{E}[|H|] < \infty$. Analogously, $U_{i\wedge T}$ also converges a.s. to a limit $U$ with $\mathbb{E}[|U|] < \infty$ as $i \rightarrow \infty$. Putting this together
    \begin{align*}
        i\wedge T &= \left| (i\wedge T) - \frac{-h(\boldsymbol{\theta};\mathcal{F}_{i\wedge T})}{L} + \frac{-h(\boldsymbol{\theta};\mathcal{F}_{i\wedge T})}{L} \right| \\
        &\leq \left| (i\wedge T) - \frac{-h(\boldsymbol{\theta};\mathcal{F}_{i\wedge T})}{L} \right|+\left|  \frac{-h(\boldsymbol{\theta};\mathcal{F}_{i\wedge T})}{L} \right| \\
        & = |U_{i\wedge T}| +  \frac{\left|-h(\boldsymbol{\theta};\mathcal{F}_{i\wedge T})\right|}{L} \xrightarrow{\text { a.s. }} |U| + \frac{|H|}{L} < \infty .
    \end{align*}
    For large enough $i$, $i\wedge T = T $, which implies $T<\infty $ a.s. and therefore $\mathbb{E}[|T|] < \infty$. Combining this fact with \eqref{eq.U_dif}, we conclude that the conditions for the optional stopping theorem are fulfilled. 
\end{proof}

\section{Details on Human Response Modeling}\label{app:human_response_model}
To model annotator response times for word selection and word ranking queries, we conducted a crowdsourced study on Prolific. The study was categorized as minimal risk research qualified for exemption status under 45 CFR 46 104d.2 by the Institutional Review Board (IRB).

\subsection{Participant Selection and Demographics}
To promote data quality, we restricted participation to Prolific users with an approval rate of at least 95\% and at least 1{,}000 prior submissions. To ensure strong English proficiency, we limited eligibility to participants located in the United Kingdom or United States who reported completing an undergraduate degree and listed English as their primary language. All eligible participants received detailed instructions on the task and the annotation interface. Within these instructions, we included multiple-choice attention check questions to verify comprehension. Twenty-two individuals did not pass the preliminary attention checks and were excluded from the study.

After applying eligibility filters and attention checks, \textbf{101 annotators} participated in our study. The demographic breakdown was as follows: 42 female and 59 male participants; 41 residing in the UK and 60 in the US. The majority of participants were students (69 yes, 17 no, 15 no response). The age distribution is shown in Figure \ref{fig:age_histogram}, with a median age of 44 years. The breakdown by ethnicity and country of birth is provided in Tables \ref{tab:ethnicity} and \ref{tab:country_birth}, respectively. Most participants self-identified as White and were born in either the UK or USA.

\begin{figure}
    \centering
    \includegraphics[width=0.8\linewidth]{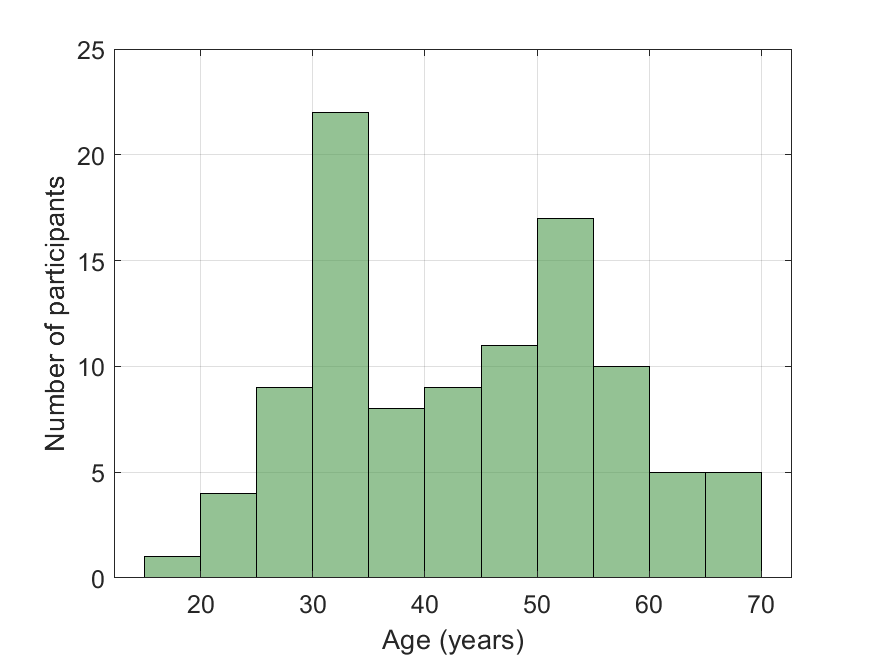}
    \caption{Age distribution of participants. The mean age is 43.25 years and median 44 years.}
    \label{fig:age_histogram}
\end{figure}

\begin{table}[htbp]
    \centering
    \caption{Ethnicity distribution}
    \label{tab:ethnicity}
    \begin{tabular}{lr}
    \toprule
    Ethnicity & Count \\
    \midrule
    White & 68 \\
    Asian & 15 \\
    Black & 11 \\
    Mixed & 4 \\
    Other & 1 \\
    Not available & 2 \\
    \midrule
    \textbf{Total} & \textbf{101} \\
    \bottomrule
    \end{tabular}
\end{table}

\begin{table}[htbp]
    \centering
    \caption{Country of birth distribution}
    \label{tab:country_birth}
    \begin{tabular}{lr}
    \toprule
    Country & Count \\
    \midrule
    United States & 55 \\
    United Kingdom & 36 \\
    Nigeria & 2 \\
    Bulgaria & 1 \\
    China & 1 \\
    Hungary & 1 \\
    Indonesia & 1 \\
    Japan & 1 \\
    Korea & 1 \\
    Malta & 1 \\
    Philippines & 1 \\
    \midrule
    \end{tabular}
\end{table}

To mitigate potential confounds caused by interface familiarity and learning effects, we followed standard A/B testing practice and counterbalanced the query order. Half of the participants completed the word-ranking queries first and then the word-selection queries, while the remaining participants completed the two query types in the reverse order.

Figure \ref{fig:totaltimeviolin} reports the total time taken by participants to go through the instructions and answer all 50 queries (25 $q_{\text{rank}}$ queries and 25 $q_{\text{high}}$ queries). On average, participants required approximately 20 minutes. Based on the observed completion times and the study compensation, the hourly wage was on average \(12.23\) USD, with a median of \(10.60\) and a standard deviation of \(9.53\).

\begin{figure}
    \centering
    \includegraphics[width=0.8\linewidth]{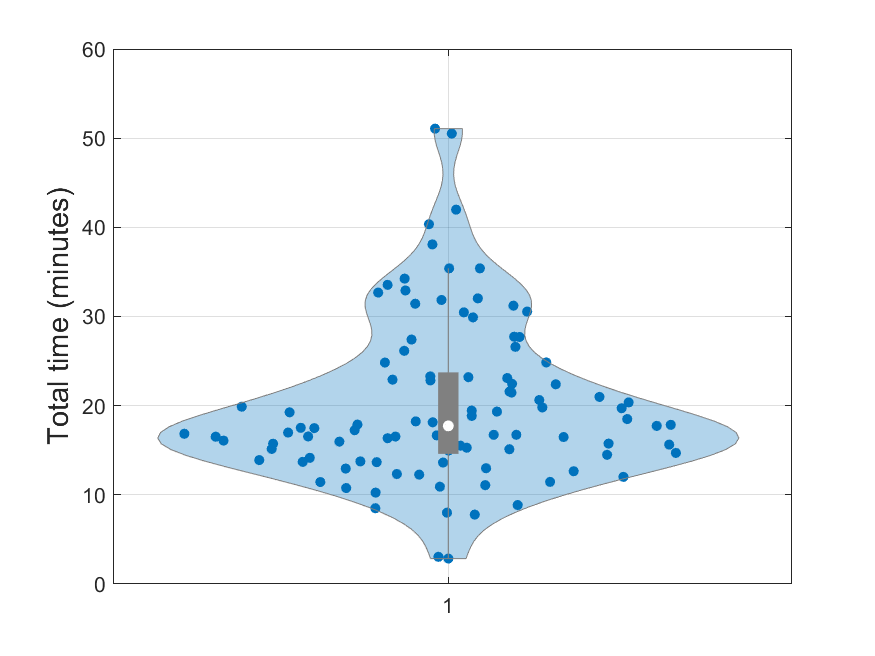}
    \caption{Distribution of total study completion time per participant (including instructions and all word-selection and word-ranking queries).}
    \label{fig:totaltimeviolin}
\end{figure}

Beyond the eligibility criteria and pre-task attention checks described above, we applied additional post hoc quality control to identify inattentive annotators before modeling response times. In particular, we embedded the following five gold-standard sentiment queries with clear and unambiguous sentiment polarity: 
\begin{itemize}
    \item \textit{amazing, rejection}
    \item \textit{bad, horrible}
    \item \textit{terrified, worried}
    \item \textit{inspiring, boring}
    \item \textit{happy, sad}
\end{itemize}
These gold-standard queries were intentionally distributed across the task at fixed positions (queries 1, 2, 11, 12, and 21), ensuring that attention was assessed at the beginning, middle, and end of the experiment rather than only at a single point. We excluded any participant who failed to exactly match the expected response on at least one gold-standard query, since these items were constructed to have an unambiguous correct answer for annotators who understood the task and responded attentively.

After applying all quality-control filters, including gold-standard checks, the final dataset contained 1648 observations for the word selection task and 1256 observations for the ranking task. The breakdown by query type is shown in Table \ref{tab:counts-N}.
\begin{table}[htbp]
\centering
\caption{Counts by word set size after preprocessing (gold-standard queries removed).}
\label{tab:counts-N}
\begin{tabular}{c rr}
\toprule
$N$ & Ranking & Selection \\
\midrule
2 & 142 & 147 \\
3 & 184 & 272 \\
4 & 175 & 269 \\
5 & 204 & 234 \\
6 & 176 & 276 \\
7 & 191 & 225 \\
8 & 184 & 225 \\
\bottomrule
\end{tabular}
\end{table}
A total of 83 participants remained for the word selection task and 63 participants remained for the ranking task. These cleaned datasets were used for all subsequent analyses.

\subsection{Data Analysis}\label{app:DataAnalysis}
Before collecting the data, we hypothesized the following models could accurately describe the response time:
\begin{itemize}
    \item \textbf{\underline{Hypothesis 1} Linear model for selection queries}: Drawing from the literature on serial scanning, we hypothesize that the response time increases linearly with the number of options: $ \widehat{t}_{\text{selection}} = \beta_0 + \beta_1 |\mathcal{S}|$.
    \item \textbf{\underline{Hypothesis 2} Logarithmic model for selection queries}: Based on Hick's law, we hypothesize that the response time increases logarithmically with the number of options: $ \widehat{t}_{\text{selection}} = \beta_0 + \beta_1 \log|\mathcal{S}|$.
    \item \textbf{\underline{Hypothesis 3} Linear model for ranking queries}: Assuming reading or viewing the options dominates the burden, we hypothesize that the response time increases linearly with the number of options: $ \widehat{t}_{\text{rank}} = \beta_0 + \beta_1 |\mathcal{S}|$.
    \item \textbf{\underline{Hypothesis 4} Quadratic model for ranking queries}: Inspired by the complexity of simple sorting algorithms like Bubble sort, we hypothesis the response time increases quadratically with the number of options: $ \widehat{t}_{\text{rank}} = \beta_0 + \beta_1 |\mathcal{S}|^2$.   
\end{itemize}
We evaluated these candidate parametric forms on the cleaned dataset. Table \ref{tab:model_comparison} presents the complete regression results for all candidate models. For both question types, the linear models achieved higher $R^2$ values and lower \ac{MSE}, confirming superior fit as concluded by the Vuong tests. Notably, the slope for ranking queries ($4.41$) is substantially steeper than for selection queries ($0.63$), indicating that each additional word imposes a much greater time burden when participants must produce a complete ranking rather than identify a single word.

\begin{table*}[t]
    \centering
    \caption{Comparison of candidate response time models fitted using least-squares regression.}
    \label{tab:model_comparison}
    \begin{tabular}{llcccccc}
    \toprule
    \textbf{Query Type} & \textbf{Model} & \textbf{$\beta_0$ (SE)} & \textbf{$\beta_1$ (SE)} & \textbf{$R^2$} & \textbf{Adjusted $R^2$} & \textbf{MSE} & \textbf{$N$} \\
    \midrule
    Selection & Linear & 4.01 (0.24) & 0.63 (0.04) & 0.112 & 0.111 & 11.22 & 1648 \\
    Selection & Logarithmic & 3.10 (0.32) & 1.84 (0.14) & 0.097 & 0.097 & 11.36 & 1648 \\
    \midrule
    Ranking & Linear & $-0.32$ (0.88) & 4.41 (0.16) & 0.373 & 0.372 & 125.44 & 1256 \\
    Ranking & Quadratic & 9.69 (0.57) & 0.42 (0.02) & 0.356 & 0.356 & 127.69  & 1256 \\
    \bottomrule
    \end{tabular}
\end{table*}

\subsection{Item difficulty stratification and effect on response time}\label{app:difficulty_time}
Curriculum learning \cite{Bengio2009CurriculumLearning} suggests that optimal learning strategies often progress from simple examples to more complex ones. We hypothesized that our information-gain-based active learning strategy might exhibit a curriculum-like progression: selecting relatively easy queries at first, and transitioning to more challenging boundary cases as the uncertainty about the classifier decreases. If true, learning stage could confound the response time model beyond the effect of set size $|\mathcal{S}|$.

To test this, we analyzed queries from three stages of the learning process: early (first 5 queries), mid-stage (around iteration 200), and late-stage (around iteration 1000). Using selection stage as a proxy for difficulty, we examined whether queries selected at different points in training exhibited different response time patterns after controlling for set size.

\begin{figure}
    \centering
    \includegraphics[width=0.8\linewidth]{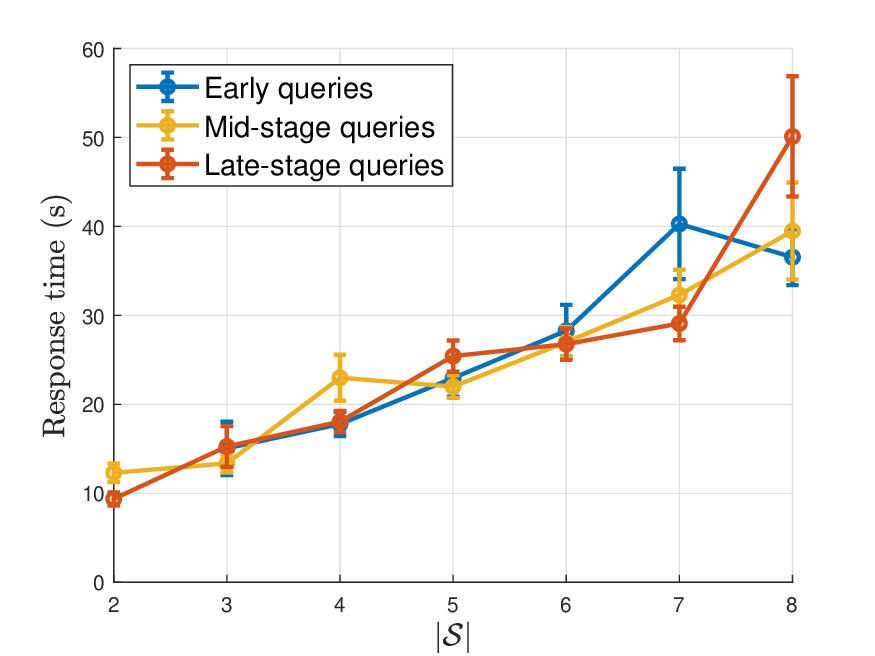}
    \caption{Query Response time versus query set size stratified by stage of the learning process. Queries were sampled from early (first 5 queries), mid (iteration ~200), and late stages (iteration ~1000). The shaded area represents the standard error. The three categories show highly overlapping trends, indicating that learning stage contributes minimal variance beyond set size. }
    \label{fig:difficulty}
\end{figure}

Figure \ref{fig:difficulty} presents the response time as a function of word set size $|\mathcal{S}|$ for these three difficulty tiers. To formally test whether difficulty explains additional variance beyond set size, we conducted an analysis of covariance (ANCOVA) comparing a baseline model containing only \( |\mathcal{S}| \) to one that also included difficulty as a categorical predictor. The extended model did not improve fit (\( p = 0.90 \)), indicating that once query length is accounted for, difficulty contributes no measurable additional variance to response time. Consequently, we model response time solely as a function of \( |\mathcal{S}| \) in the main analysis.
\end{document}

%% file: Acronyms.tex
\acrodef{ACDIS}[ACDIS]{Adaptive Communication Decision and Information Systems}
\acrodef{AL}[AL]{Active Learning}
\acrodef{MMSE}[MMSE]{Minimum Mean-Square Error}
\acrodef{MSE}[MSE]{Mean-Square Error}
\acrodef{KL}[KL]{Kullback-Leibler}
\acrodef{BBVI}[BBVI]{Black Box Variational Inference}
\acrodef{LLM}[LLM]{Large Language Model}
\acrodef{LCC}[LCC]{log-concave}
\acrodef{ELBO}[ELBO]{log Evidence Lower BOund}
\acrodef{NRC}[NRC]{National Research Council Canada}